\documentclass[11pt,USenglish]{article}

\usepackage{fullpage}
\usepackage[T1]{fontenc}
\usepackage{authblk}
\usepackage[pagebackref,colorlinks,linkcolor=red!50!black,citecolor=green!50!black]{hyperref}
\usepackage{amsthm}
\usepackage{amsmath}
\usepackage{amssymb}
\usepackage{subcaption}
\usepackage{amsfonts}
\usepackage{dsfont}
\usepackage{subcaption}
\usepackage[font=small,labelfont=bf]{caption}
\usepackage{thmtools}
\usepackage{thm-restate}
\usepackage[linesnumbered,boxruled]{algorithm2e}
\usepackage{mathtools}
\usepackage{paralist}
\usepackage[nameinlink]{cleveref}

\usepackage{booktabs}
\usepackage{tabularx}
\usepackage{multirow}
\usepackage{makecell}

\usepackage{enumitem}
\usepackage[sort,numbers]{natbib} %
\usepackage{tikz}
\usetikzlibrary{arrows}
\usetikzlibrary{shapes}
\usetikzlibrary{decorations.pathreplacing}

\usepackage{bm}
\usepackage{etoolbox}

\declaretheorem[name=Theorem,numberwithin=section]{theorem}

\newtheorem{observation}[theorem]{Observation}

\newtheorem{claim}[theorem]{Claim}
\newtheorem{rrule}{Reduction Rule}[section]

\crefname{rrule}{Reduction Rule}{Reduction Rules}
\crefname{claim}{Claim}{Claims}
\crefname{paragraph}{Paragraph}{Paragraphs}
\crefname{observation}{Observation}{Observations}
\crefname{lemma}{Lemma}{Lemmata}
\crefname{lemma}{Lemma}{Lemmata}
\crefname{theorem}{Theorem}{Theorems}
\crefname{proposition}{Proposition}{Propositions}
\crefname{corollary}{Corollary}{Corollaries}
\crefname{remark}{Remark}{Remarks}
\crefname{section}{Section}{sections}
\crefname{figure}{Figure}{Figures}
\crefname{table}{Table}{Tables}
\crefname{definition}{Definition}{Definitions}
\crefname{equation}{Equation}{Equations}
\crefname{algorithm}{Algorithm}{Algorithms}

\Crefname{remark}{Remark}{Remarks}
\Crefname{paragraph}{Paragraph}{Paragraphs}
\Crefname{theorem}{Thm}{Thms}
\Crefname{proposition}{Prop}{Props}
\Crefname{corollary}{Cor}{Cors}
\Crefname{observation}{Obs}{Obs}

\theoremstyle{definition}
\newtheorem{definition}[theorem]{Definition}
\theoremstyle{remark}

\theoremstyle{plain}

\makeatletter
\def\thmt@refnamewithcomma #1#2#3,#4,#5\@nil{%
  \@xa\def\csname\thmt@envname #1utorefname\endcsname{#3}%
  \ifcsname #2refname\endcsname
    \csname #2refname\expandafter\endcsname\expandafter{\thmt@envname}{#3}{#4}%
  \fi
}
\makeatother

\definecolor{ourblue}{RGB}{135,206,250}
\definecolor{ourgreen}{RGB}{0,100,0}
\definecolor{ourred}{RGB}{176,23,31}
\definecolor{lipicsyellow}{rgb}{0.99,0.78,0.07}

\newcommand{\appendixproof}[3]{%
	#3
}

\newcommand{\appendixsection}[1]{%
}

\usepackage{amsopn}

\newcommand{\FPT}{\textnormal{\textsf{FPT}}}
\newcommand{\Wone}{\textnormal{\textsf{W[1]}}}

\newcommand{\MSO}{\textnormal{\textsf{MSO}}}
\newcommand{\N}{\mathds N}

\newcommand{\I}{\mathcal{I}}
\newcommand{\TG}{\boldsymbol{G}}
\newcommand{\TE}{\boldsymbol{E}}

\newcommand{\NP}{\textnormal{\textsf{NP}}}

\newcommand{\ON}{\mathcal{O}}
\newcommand{\aq}{\Leftrightarrow}
\newcommand{\impl}{\Rightarrow}
\newcommand{\yes}{\textnormal{\textsf{yes}}}
\newcommand{\no}{\textnormal{\textsf{no}}}
\newcommand{\ug}[1]{{#1}_{\downarrow}}

\newcommand{\nonstrproblem}{\textnormal{\textsc{Temporal $(s,z)$-Sep\-a\-ra\-tion}}}

\newcommand{\strproblem}{\textnormal{\textsc{Strict Temporal $(s,z)$-Sep\-a\-ra\-tion}}}
\newcommand{\sproblem}{\textnormal{\textsc{(Strict) Temporal $(s,z)$-\allowbreak Sep\-a\-ra\-tion}}}
\newcommand{\lenbndproblem}{\textnormal{\textsc{Length-Bound\-ed $(s,z)$-Sep\-a\-ra\-tion}}}
\newcommand{\lenbndcutproblem}{\textnormal{\textsc{Length-Bound\-ed $(s,z)$-Cut}}}
\newcommand{\vertexcover}{\textnormal{\textsc{Vertex Cover}}}

\newcommand{\nonstrpath}[1]{temporal~$(#1)$-path}
\newcommand{\nonstrpaths}[1]{temporal~$(#1)$-paths}
\newcommand{\strpath}[1]{strict temporal~$(#1)$-path}
\newcommand{\strpaths}[1]{strict temporal~$(#1)$-paths}

\newcommand{\npath}[1]{$(#1)$-path}
\newcommand{\npaths}[1]{$(#1)$-paths}

\newcommand{\nonstrsep}[1]{temporal~$(#1)$-sep\-a\-ra\-tor}

\newcommand{\strsep}[1]{strict temporal~$(#1)$-sep\-a\-ra\-tor}

\newcommand{\nsep}[1]{$(#1)$-sep\-a\-ra\-tor}

\newcommand{\implone}{\smallskip\noindent$\Rightarrow$:}
\newcommand{\impltwo}{\smallskip\noindent$\Leftarrow$:}

\newcommand{\appsymb}{$\star$}

\usepackage{xparse}
\usepackage{mathtools}
\DeclareMathOperator{\tw}{tw}

\DeclarePairedDelimiterX{\set}[1]{\{}{\}}{\setargs{#1}}
\NewDocumentCommand{\setargs}{>{\SplitArgument{1}{;}}m}
{\setargsaux#1}
\NewDocumentCommand{\setargsaux}{mm}
{\IfNoValueTF{#2}{#1} {#1\,\delimsize|\,\mathopen{}#2}}

\newdimen\longformulasindent

\newcommand{\multiwaycut}{\textnormal{\textsc{Node Multiway Cut}}}
\newcommand{\multiwaycutAcr}{\textnormal{\textsc{NWC}}}
\newcommand{\condRef}[1]{{\hyperref[#1]{(\ref{#1})}}}
\newcommand{\indic}{\ensuremath{\mathds{1}}}
\newcommand{\mte}{\bm{m}}

\newtheorem{lem}[theorem]{Lemma}
\makeatletter
\@addtoreset{lem}{section}
\makeatother

\newtheorem{cor}[theorem]{Corollary}

\crefname{rrule}{Reduction Rule}{Reduction Rules}
\crefname{claim}{Claim}{Claims}
\crefname{paragraph}{Paragraph}{Paragraphs}
\crefname{observation}{Observation}{Observations}
\crefname{lem}{Lemma}{Lemmata}
\crefname{theorem}{Theorem}{Theorems}
\crefname{proposition}{Proposition}{Propositions}
\crefname{corollary}{Corollary}{Corollaries}
\crefname{remark}{Remark}{Remarks}
\crefname{section}{Section}{sections}
\crefname{figure}{Figure}{Figures}
\crefname{table}{Table}{Tables}
\crefname{equation}{Equation}{Equations}
\crefname{algorithm}{Algorithm}{Algorithms}

\Crefname{remark}{Remark}{Remarks}
\Crefname{paragraph}{Paragraph}{Paragraphs}
\Crefname{theorem}{Thm.}{Thms.}
\Crefname{proposition}{Prop.}{Props.}
\Crefname{corollary}{Cor.}{Cors.}
\Crefname{cor}{Cor.}{Cors}
\Crefname{observation}{Obs.}{Obs.}
\Crefname{section}{Sec.}{Sec.}
\theoremstyle{definition}
\newtheorem{defin}[theorem]{Definition}
\crefname{defin}{Definition}{Definitions}
\theoremstyle{remark}

\newcommand{\problemdef}[3]{
    \begin{center}
    \begin{minipage}{0.95\textwidth}
      \noindent
      \normalsize\textsc{#1}
      
      \vspace{1pt}
      \setlength{\tabcolsep}{3pt}
      \renewcommand{\arraystretch}{1.0}
      \begin{tabularx}{\textwidth}{@{}lX@{}}
	\normalsize\textbf{Input:} 	& \normalsize#2 \\
	\normalsize\textbf{Question:} 	& \normalsize#3
      \end{tabularx}
    \end{minipage}
    \end{center}
}

\usepackage{mathrsfs}
\usepackage{yfonts}

\newcommand{\thetitle}{The Complexity of Finding Small Separators in Temporal Graphs }
\title{\thetitle{}
}

\author{Philipp~Zschoche\thanks{Supported by the Stiftung Begabtenf\"orderung berufliche Bildung (SBB).}}
\author{Till~Fluschnik\thanks{Supported by the DFG, project DAMM (NI 369/13) and project TORE (NI 369/18).}}
\author{Hendrik~Molter\thanks{Supported by DFG, project MATE (NI 369/17).}} 
\author{Rolf~Niedermeier}

\affil{Institut f\"ur Softwaretechnik und Theoretische Informatik,
 TU Berlin, Germany,\\
 \texttt{\{zschoche,till.fluschnik,h.molter,rolf.niedermeier\}@tu-berlin.de}}
\usepackage{mathrsfs}
\usepackage{yfonts}
\date{}

\begin{document}
\maketitle

%

\begin{abstract}
Temporal graphs are graphs with time-stamped edges.
We study the problem of finding a small vertex set (the separator) with respect to two designated terminal vertices such that the removal of the set eliminates all temporal paths connecting one terminal to the other.
Herein, we consider two models of temporal paths: paths that pass through arbitrarily many edges per time~step (non-strict) and paths that pass through at most one edge per time~step (strict).
Regarding the number of time~steps of a temporal graph, we show a complexity dichotomy (\textsf{NP}-hardness versus polynomial-time solvability) for both problem variants.
Moreover we prove both problem variants to be \textsf{NP}-complete even on temporal graphs whose underlying graph is planar.
We further show that, on temporal graphs with planar underlying graph, if additionally the number of time~steps is constant, then the problem variant for strict paths is solvable in quasi-linear time.
Finally, we introduce and motivate the notion of a temporal core (vertices whose incident edges change over time).
We prove that the non-strict variant is fixed-parameter tractable when parameterized by the size of the temporal core, while the strict variant remains~\textsf{NP}-complete, even for constant-size temporal cores.
\end{abstract}

\section{Introduction}
  \label{sec:intro}

In complex network analysis, it is nowadays very common to have access to and process graph data where the interactions among the vertices are time-stamped. 
When using static graphs as a mathematical model, the dynamics of interactions are not reflected and important information of the data might not be captured.
\emph{Temporal graphs} address this issue.
A temporal graph is, informally speaking, a graph where the edge set may change over a discrete time interval, while the vertex set remains unchanged.
%
Having the dynamics of interactions represented in the model, it is essential to adapt definitions such as connectivity and paths to respect temporal features. 
This directly affects the notion of \emph{separators} in the temporal setting. 
Vertex separators are a fundamental primitive in static network analysis and it is well-known that they can be computed in polynomial time~(see, e.g., proof of~\cite[Theorem~6.8]{AMO93}).
In contrast to the static case, \citet{kempe2000connectivity} showed that in temporal graphs it is \NP-hard to compute minimum separators.

Temporal graphs are well-established in the literature and are also referred to as time-varying~\cite{liang2017survivability} and evolving~\cite{ferreira2004building} graphs, temporal networks~\cite{holme2012temporal, kempe2000connectivity, mertzios2013temporal},  link streams~\cite{latapy2017stream,viard2015revealing}, multidimensional networks~\cite{boccaletti2014structure}, and edge-scheduled networks~\cite{berman}.
In this work, we use the well-established model in which each edge has a time stamp~\cite{boccaletti2014structure, holme2012temporal, akrida2015temporally, himmel2016enumerating, kempe2000connectivity, mertzios2013temporal,AMSZ18}. 
Assuming discrete time steps, this is equivalent to a sequence of static graphs over a fixed set of vertices~\cite{michail2016introduction}. 
Formally, we define a temporal graph as follows.
\begin{definition}[Temporal Graph]
	An (undirected) \emph{temporal graph}~$\TG = (V,\TE,\tau)$ is an ordered triple consisting of a set~$V$ of vertices,
	a set~$\TE \subseteq \binom{V}{2} \times \{1,\dots, \tau\}$ of \emph{time-edges},
	and a maximal time label~$\tau \in \N$.
\end{definition}
See \cref{fig:example} for an example with $\tau =4$,
that is, a temporal graph with four time steps, also referred to as \emph{layers}. 
The static graph obtained from a temporal graph~$\TG$ by removing the time stamps from all time-edges we call the \emph{underlying graph} 
of~$\TG$.

\begin{figure}
			\centering
	\subcaptionbox{\centering A temporal graph $\TG$.}[0.3\textwidth]{\includegraphics[width=0.3\textwidth]{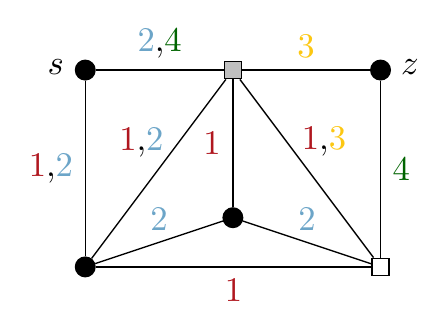}}
	\hfill
			\centering
	\subcaptionbox{\centering Layers of $\TG$.}[0.65\textwidth]{\includegraphics[width=0.65\textwidth]{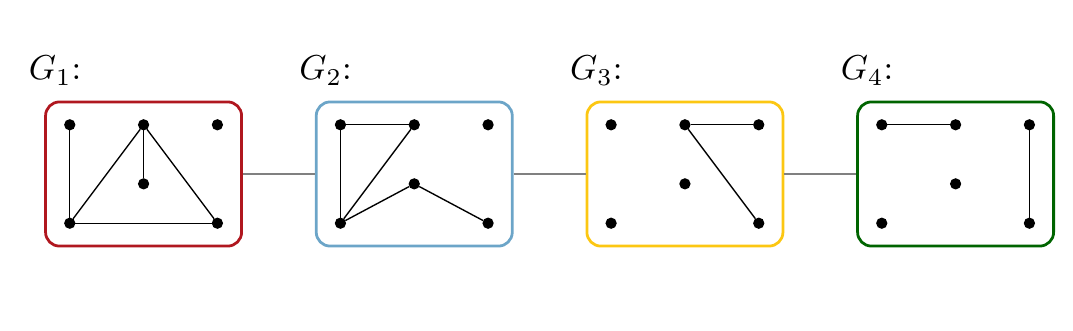}}

	\centering
	\caption{
		Subfigure (a) shows a temporal graph~$\TG$ and subfigure (b) shows its four layers~$G_1,\ldots,G_4$.
		The gray squared vertex forms a \strsep{s,z}, but no \nonstrsep{s,z}. 
		The two squared vertices form a \nonstrsep{s,z}.
		}
		\label{fig:example}
\end{figure}

Many real-world applications have temporal graphs as underlying mathematical model. 
For instance, it is natural to model connections in public transportation networks with temporal graphs. 
Other examples include information spreading in social networks, communication in social networks, biological pathways, or spread of diseases \cite{holme2012temporal}.
 
A fundamental question in temporal graphs, addressing issues such as 
connectivity~\cite{AxiotisF16,mertzios2013temporal}, survivability~\cite{liang2017survivability},
and robustness~\cite{ScellatoLMBZ13},
is whether there is a ``time-respecting'' path from a 
distinguished start vertex~$s$ to a distinguished target 
vertex~$z$.\footnote{In the literature the sink is usually denoted by $t$.
To be consistent with \citet{michail2016introduction} we use $z$ instead as we reserve $t$ to refer to points in time.}
We provide a thorough study of the computational complexity of separating 
$s$ from~$z$ in a given temporal graph. 

Moreover, we study two natural restrictions of temporal graphs:
\begin{inparaenum}[(i)]
\item planar temporal graphs and
\item temporal graphs with a bounded number of vertices incident to edges that are not permanently existing---these vertices form the so-called \emph{temporal core}.
\end{inparaenum}
Both restrictions are naturally motivated by settings e.g.\ occurring in (hierarchical) traffic networks.
We also consider two very similar but still significantly differing temporal 
path models (both used in the literature), leading to two corresponding 
models of temporal separation.

\subparagraph{Two path models.}
We start with the introduction of the ``non-strict'' path model~\cite{kempe2000connectivity}.
Given a temporal graph~$\TG=(V,\TE,\tau)$ with two distinct vertices~$s,z\in V$, a 
\emph{temporal~$(s,z)$-path} of 
length~$\ell$ in~$\TG$ is a sequence
$P = ( (\{s=v_0,v_1\}, t_1), (\{v_1, v_2\}, t_2), \dots, (\{v_{\ell-1}, v_\ell=z\}, t_{\ell} ) )$  of time-edges in~$\TE$,
where $v_i\neq v_j$ for all $i, j\in \{0,\ldots,\ell\}$ with $i\neq j$
and~$t_i \leq t_{i+1}$ for all~$i \in \{1, \ldots,\ell -1\}$.
A vertex set~$S$ with $S\cap\{s,z\}=\emptyset$ is a \emph{\nonstrsep{s,z}} if there is no \nonstrpath{s,z} in~$\TG - S:= (V \setminus S, \set{ (\{v,w\},t) \in \TE ; v,w \in V \setminus S}, \tau)$. 
We are ready to state the central problem of our paper.
\problemdef{\nonstrproblem{}}
		{ A temporal graph~$\TG=(V,\TE,\tau)$, two distinct vertices~$s,z\in V$, and~$k \in \N$.}
		{Does $\TG$ admit a \nonstrsep{s,z} of size at most~$k$?   }
Our second path model is the ``strict'' variant.
A \nonstrpath{s,z} $P$ is called \emph{strict} if $t_i < t_{i+1}$ for all $i \in \{ 1, \dots, \ell -1 \}$.
In the literature, strict temporal paths are also known as journeys \cite{akrida2015temporally,akrida2017temporal,
michail2016introduction,mertzios2013temporal}.\footnote{We also refer to \citet{Himmel18} for a thorough discussion and  comparison of temporal path concepts.}
A vertex set~$S$ is a \emph{\strsep{s,z}} if there is no \strpath{s,z} in~$\TG - S$.
Thus, our second main problem, \strproblem{}, 
is defined in complete analogy to \nonstrproblem{}, just replacing (non-strict) 
temporal separators by strict ones.

While the strict version of temporal separation immediately appears as natural, the non-strict variant can be viewed as a more conservative version of the problem. 
For instance, in a disease-spreading scenario the spreading speed might be unclear.
To ensure containment of the spreading by separating patient zero ($s$) from a certain target ($z$), a \nonstrsep{s,z} might be the safer choice.

\subparagraph{Main results.}
\newcommand{\smtab}[1]{\scriptsize#1}
\newcommand{\mrrb}[2]{\multirow{#1}{*}{\rotatebox[origin=c]{90}{#2}}}
\renewcommand{\arraystretch}{1.1}
\begin{table}[t]
  \setlength{\tabcolsep}{5pt}
  \centering
  \caption
  {
	  Overview on our results. Herein, \NP{}-c.~abbreviates \NP{}-complete,~$n$ and~$\mte$ denote the number of vertices and time-edges, respectively,~$\ug{G}$ refers to the underlying graph of an input temporal graph.
    $ ^a$\,(\Cref{thm:nphard}; \Wone-hard wrt.~$k$) $^b$\,(\Cref{thm:str-t-4}) $^c$\,(\Cref{cor:sprobhardonplanar}) $^d$\,(\Cref{thm:planarfpt}) $^e$\,(\Cref{thm:fptcore})
  }
  \def\fbs{3.5cm}
  \begin{tabular}{@{}llllll@{}}%
  \toprule
  & \multicolumn{2}{l}{General} & \multicolumn{2}{l}{Planar~$\ug{G}$} & Temporal core  
  \vspace{-4pt}\\ &  \multicolumn{2}{l}{\smtab{(\cref{sec:gh})}} &  \multicolumn{2}{l}{\smtab{(\cref{sec:nonstr})}} & \smtab{(\cref{sec:strict})}
  \\\cmidrule{2-6}
	$(s,z)$-\textsc{Separation}	& $2\leq \tau\leq 4$	& $5\leq \tau$	& $\tau$~unbounded	& $\tau$~constant & constant size
	\\
    \midrule
      \textsc{Temporal}	& \multicolumn{2}{c}{\NP{}-complete$ ^a$}
					    & \NP{}-c.$^c$
					    & \emph{open}
					    & $n^{O(1)}+O(\mte\log \mte)$~$^e$
					    \\
	  \textsc{Strict Temporal}	& $\ON(k\cdot \mte)$~$^b$ & \NP{}-c.$ ^a$ & \NP{}-c.$^c$ & $\ON(\mte\log \mte)$~$^d$ & \NP{}-complete$^a$
							\\
    \bottomrule
  \end{tabular}
  \label{tab:results}
\end{table}
%
\cref{tab:results} provides an overview on our results.

A central contribution is to prove that both \nonstrproblem{} and \strproblem{} are \NP-complete for all~$\tau\geq 2$ and~$\tau\geq 5$, respectively, strengthening a result by~\citet{kempe2000connectivity} (they show \NP-hardness of both variants for all $\tau\geq 12$).
For \nonstrproblem{}, our hardness result is already tight.\footnote{\nonstrproblem{} with~$\tau=1$ is equivalent to \textsc{$(s,z)$-Separation} on static graphs.}
 For the strict variant, we identify a dichotomy in the computational complexity by proving polynomial-time solvability of \strproblem{} for~$\tau\leq 4$.
Moreover, we prove that both problems remain \NP-complete on temporal graphs that have an underlying graph that is planar.

 

 We introduce the notion of temporal cores in temporal graphs.
 Informally, the temporal core of a temporal graph is the set of vertices whose edge-incidences change over time.
We prove that \nonstrproblem{} is fixed-parameter tractable (\FPT) when parameterized by the size of the temporal core, while~\strproblem{} remains~\NP-complete even if the temporal core is empty.

A particular aspect of our results is that they demonstrate that the choice of the model (strict versus non-strict) for a problem can have a crucial impact on the computational complexity of said problem.
This contrasts with wide parts of the literature where both models were used without discussing the subtle but crucial differences in computational complexity.

\subparagraph{Technical contributions.}
To show the polynomial-time solvability of \strproblem{} for~$\tau\leq 4$, we prove that a classic separator result of \citet{lovasz1978mengerian} translates to the strict temporal setting. 
This is surprising since many other results about separators in the static case do not apply in the temporal case. In this context, we also develop a linear-time algorithm for \textsc{Single-Source Shortest Strict Temporal Paths}, improving the running time of the best known algorithm due to~\citet{wu2016efficient} by a logarithmic factor.

We settle the complexity of \lenbndproblem{} 
on planar graphs by showing its \NP-hardness, which was left unanswered by~\citet{FHNN18} and promises to be a valuable intermediate problem for proving hardness results.
In the hardness reduction for \lenbndproblem{} we introduce a grid-like, planarity-preserving vertex gadget that is generally useful to replace ``twin'' vertices which in many cases are not planarity-preserving and which are often used to model weights. 

While showing that \nonstrproblem{} is fixed-parameter tractable when parameterized by the size of the temporal core, we employ a case distinction on the size of the temporal core, and show that in the non-trivial case we can reduce the problem to \multiwaycut{}. 
We identify an ``above lower bound parameter'' for \multiwaycut{} that is suitable to lower-bound the size of the temporal core, thereby making it possible to exploit a fixed-parameter tractability result due to \citet{cygan2013multiway}. 

\subparagraph{Related work.}
Our most important reference is the work of \citet{kempe2000connectivity} who proved that \nonstrproblem{} is \NP-hard.
In contrast, \citet{berman} proved that computing temporal $(s,z)$-cuts (edge deletion instead of vertex deletion) is polynomial-time solvable. 
In the context of survivability of temporal graphs,~\citet{liang2017survivability} studied cuts where an edge deletion only lasts for~$\delta$ consecutive time stamps.
Moreover, they studied a temporal maximum flow defined as the maximum number of sets of journeys where each two journeys in a set do not use a temporal edge within some~$\delta$ time steps.
A different notion of temporal flows on temporal graphs was introduced by~\citet{akrida2017temporal}.
They showed how to compute in polynomial~time the maximum amount of flow passing from a source vertex~$s$ to a sink vertex~$z$ until a given point in time.

The vertex-variant of Menger's Theorem~\cite{Menger1927} states that the maximum number of vertex-disjoint paths from~$s$ to~$z$ equals the size of a minimum-cardinality \nsep{s,z}.
In static graphs, Menger's Theorem allows for finding a minimum-cardinality~\nsep{s,z} via maximum flow computations.
However, \citet{berman} proved that the vertex-variant of an analogue to Menger's Theorem for temporal graphs, asking for the maximum number of (strict) temporal paths instead, does not hold. 
\citet{kempe2000connectivity} proved that the vertex-variant of the former analogue
to Menger's Theorem holds true if the underlying graph excludes a fixed minor.
\citet{mertzios2013temporal} proved another analogue of Menger's Theorem: 
the maximum number of \strpath{s,z} which never leave the same vertex at the same time equals the minimum number of node departure times needed to separate~$s$ from~$z$, where a node departure time $(v,t)$ is the vertex $v$ at time point $t$. 

\citet{MichailS16} introduced the time-analogue of the famous \textsc{Traveling Salesperson} problem and studied the problem on temporal graphs of dynamic diameter~$d\in\N$, that is, informally speaking, on temporal graphs where every two vertices can reach each other in at most~$d$ time steps at any time.
\citet{erlebach2015temporal} studied the same problem on temporal graphs where the underlying graph has bounded degree, bounded treewidth, or is planar.
Additionally, they introduced a class of temporal graphs with regularly present edges, that is, temporal graphs where each edge is associated with two integers upper- and lower-bounding consecutive time steps of edge absence.
\citet{AxiotisF16} studied the problem of finding the smallest temporal subgraph of a temporal graph such that single-source temporal connectivity is preserved on temporal graphs where the underlying graph has bounded treewidth.
In companion work, we recently studied the computational complexity of (non-strict) temporal separation on several other restricted temporal graphs~\cite{FMNZ18}.


\section{Preliminaries}\label{sec:prem}
\appendixsection{sec:prem}

Let $\N$ denote the natural numbers without zero.
For~$n\in \N$, we use~$[n]:=[1,n]=\{1,\ldots,n\}$.

\smallskip\noindent\emph{Static graphs.}
We use basic notations from (static) graph theory~\cite{diestel2000graphentheory}.
Let~$G = (V, E)$ be an \emph{undirected, simple graph}. 
We use~$V(G)$ and $E(G)$ to denote the set of vertices and set of edges of~$G$, respectively.
We denote by~$G - V'$ $:=(V\setminus V',\set{ \{v,w\} \in E; v,w \in V \setminus V' })$ 
the graph~$G$ without the vertices in~$V'\subseteq V$.
For~$V' \subseteq V$,~$G[V']:=G - (V \setminus V')$ denotes the \emph{induced subgraph} of~$G$ by~$V'$. 
A \emph{path} of length~$\ell$ is sequence of edges~$P = (\{v_1, v_2\}, \{v_2, v_3\},\dots,\{v_\ell, v_{\ell+1}\})$ where $v_i\neq v_j$ for all $i, j\in[\ell+1]$ with $i\neq j$.
We set~$V(P) = \{v_1,v_2,\ldots,v_{\ell+1}\}$.
Path~$P$ is an \emph{$(s,z)$-path} if~$s=v_1$ and~$z=v_{\ell+1}$. 
A set~$S \subseteq V\setminus\{s,z\}$ of vertices is an \emph{\nsep{s,z}} if there is no \npath{s,z} in~$G-S$.

\smallskip\noindent\emph{Temporal graphs.}
Let~$\TG = (V,\TE,\tau)$ be a temporal graph.
The graph~$G_i(\TG) = (V, E_i(\TG))$ is called \emph{layer}~$i$ of the temporal graph~$\TG = (V,\TE,\tau)$ where~$\{v,w\} \in E_i(\TG) \aq (\{ v, w \}, i) \in \TE$. 
The \emph{underlying graph}~$\ug{G}(\TG)$ of a temporal graph~$\TG = (V,\TE,\tau)$ is defined as~$\ug{G}(\TG) := (V,\ug{E}(\TG))$, where~$\ug{E}(\TG) = \set{ e ; (e,t) \in \TE}$.
(We write~$G_i$,~$E_i$,~$\ug{G}$, and~$\ug{E}$ for short if~$\TG$ is clear from the context.)
For~$X \subseteq V$ we define the \emph{induced temporal subgraph} of~$X$ by~$\TG[X] := (X,\set{ (\{v,w\},t) \in \TE ; v,w \in X },\tau)$.
We say that~$\TG$ is \emph{connected} if its underlying graph $\ug{G}$ is connected.
For surveys concerning temporal graphs we refer to~\cite{casteigts2012time,michail2016introduction,holme2012temporal,latapy2017stream,holme2015}.
	  
Regarding our two models, we have the following connection:
\begin{restatable}{lem}{strtononstr}
	\label{lemma:from-str-to-nonstr}
	There is a linear-time computable many-one reduction from \strproblem{} to \nonstrproblem{} that maps any instance $(\TG=(V,\TE,\tau),s,z,k)$ to an instance $(\TG'=(V',\TE',\tau'),s,z,k')$ with~$k'=k$ and~$\tau'=2\cdot\tau$.
\end{restatable}
\begin{proof}
  Let $\I=(\TG=(V,\TE,\tau),s,z,k)$ be an instance  of \strproblem{}.
  We construct an equivalent instance $\I'=(\TG'=(V',\TE',\tau'=2\tau),s,z,k'=k)$ in linear-time.
  Set~$V'= V \cup V_E$, where~$V_E:=\set{e_{(v,w),i},e_{(w,v),i}; e \in \TE}$ is called the set of \emph{edge-vertices}.
  Next, let~$\TE'$ be initially empty.
  For each $(\{v,w\},t) \in \TE$, add the time-edges $(\{ v, e_{(v,w),t} \}, 2t-1),(\{e_{(v,w),t},w\},2t), (\{w,e_{(w,v),t}\},2t-1),(\{e_{(w,v),t},v\},2t)$ to~$\TE'$.
  This completes the construction of~$\TG'$.
  Note that this can be done in~$\ON(|V|+|\TE|)$ time.
  It holds that~$|V'| = |V|+2\cdot|\TE|$ and that~$|\TE'| = 4\cdot |\TE|$.
  
  We claim that $\I$ is a \yes-instance if and only if $\I'$ is a \yes-instance.
  
  \implone{}
  Let~$S$ be a temporal~$(s,z)$-separator in~$\TG$ of size at most~$k$.
  We claim that~$S$ is also a temporal~$(s,z)$-separator in~$\TG'$.
  Suppose towards a contradiction that this is not the case.
  Then there is a temporal~$(s,z)$-path~$P$ in~$\TG'-S$.
  Note that the vertices on~$P$ alternated between vertices in~$V$ and~$V_E$.
  As each vertex in~$V_E$ corresponds to an edge, there is a temporal $(s,z)$-path in~$\TG-S$ induced by the vertices of~$V(P)\cap V$.
  This is a contradiction.

  \impltwo{}
  Observe that from any temporal~$(s,z)$-separator, we can obtain a temporal~$(s,z)$-separator of not larger size that only contains vertices in~$V$.
  Let~$S'$ be a temporal~$(s,z)$-separator in~$\TG'$ of size at most~$k$ only containing vertices in~$S$.
  We claim that~$S'$ is also a temporal~$(s,z)$-separator in~$\TG'$.
  Suppose towards a contradiction that this is not the case.
  Then there is a temporal~$(s,z)$~path~$P$ in~$\TG-S'$.
  Note that we can obtain a temporal~$(s,z)$-path~$P'$ in~$\TG'-S'$ by adding for all consecutive vertices $v$, $w$, where~$v$ appears before~$w$ at time-step~$t$ on~$P$, the vertex~$e_{(v,w),t}$.
  This is a contradiction.
\end{proof}
Throughout the paper we assume that the underlying graph of the temporal input graph $\TG$ is connected and that there is no time-edge between~$s$ and~$z$.
Furthermore, in accordance with \citet{wu2016efficient} we assume that the time-edge set $\TE$ is ordered by ascending time stamps.
%

 \subsection{The Maximum Label is Bounded in the Input Size}
 In the following, we prove that for every temporal graph in an input to~\sproblem{}, we can assume that the number of layers is at most the number of time-edges.
Observe that a layer of a temporal graph that contains no edge is irrelevant for~\nonstrproblem{}.
This also holds true for the strict case.
Hence, we can delete such a layer from the temporal graph.
This observation is formalized in the following two data reduction rules.
\begin{rrule}\label{rr:remove-idle}
	Let~$\TG = (V,\TE,\tau)$ be a temporal graph 
	and let~$[t_1, t_2] \subseteq [1,\tau]$ be an interval where for all~$t \in [t_1, t_2]$ the layer~$G_{t}$ is an edgeless graph.
	Then for all~$(\{v,w\},t') \in \TE$ where~$t' > t_2$ replace~$(\{v,w\},t')$ with~$(\{v,w\},t' - t_2 + t_1 - 1)$ in~$\TE$.
\end{rrule}
\begin{rrule}\label{rr:remove-idle-end}
	Let~$\TG = (V,\TE,\tau)$ be a temporal graph.
	If there is a non-empty interval~$[t_1, \tau]$ where  for all~$t' \in [t_1, \tau]$ the layer~$G_{t'}$ is an edgeless graph, then set~$\tau$ to~$t_1-1$. 
\end{rrule}
\noindent
We prove next that both reduction rules are exhaustively applicable in linear time.
\begin{lem}
	\label{rr-proof:remove-idle}
	\cref{rr:remove-idle,rr:remove-idle-end} do not remove or add any \nonstrpath{s,z} from/to the temporal graph~$\TG=(V,\TE,\tau)$ and can be exhaustively applied in~$\ON(|\TE|)$ time.
\end{lem}
\begin{proof}
	First we discuss \cref{rr:remove-idle}.
	Let~$\TG = (V,\TE,\tau)$ be a temporal graph,
	$s,z \in V$, 
	$[t_\alpha,t_\beta] \subseteq [\tau]$ be an interval where for all~$t \in [t_\alpha,t_\beta]$ the layer~$G_t$ is an edgeless graph.
	Let~$P = ((e_1,t_1),\dots,(e_i,t_i),(e_{i+1},t_j),\dots,(e_n,t_n))$ be a \nonstrpath{s,z},
	and let~$\TG'$ be the graph after we applied \cref{rr:remove-idle} once on~$\TG$.
	We distinguish three cases.
	\begin{compactenum}[\bf {Case }1:]
		\item If~$t_\beta > t_n$, then no time-edge of~$P$ is touched by \cref{rr:remove-idle}.
			Hence,~$P$ also exists in~$\TG'$.
		\item If~$t_i < t_\alpha < t_\beta < t_{i+1}$, then there is a \nonstrpath{s,z}~$((e_1,t_1),\dots,(e_i,t_i),$ $(e_{i+1},t_j - t_\beta + t_\alpha - 1),\dots,(e_n,t_n - t_\beta + t_\alpha - 1))$ in~$\TG'$, because~$t_i < t_{i+1} - t_\beta + t_\alpha - 1$.
		\item If~$t_\beta < t_1$, then there is clearly a \nonstrpath{s,z}~$((e_1,t_1- t_\beta + t_\alpha - 1),\dots,$ $(e_n,t_n - t_\beta + t_\alpha - 1))$ in~$\TG'$
	\end{compactenum}
	The other direction works analogously.
	We look at a \nonstrpath{s,z} in~$\TG'$ and compute the corresponding \nonstrpath{s,z} in~$\TG$.

\begin{figure}[t!]
	\begin{subfigure}[c]{0.5\textwidth}
	\centering
	\includegraphics{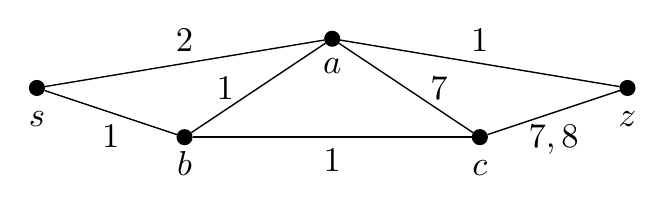}
		\subcaption{\cref{rr:remove-idle} is applicable}
		\label{fig:rr-before}
	\end{subfigure}
	\begin{subfigure}[c]{0.5\textwidth}
	\centering
	\includegraphics{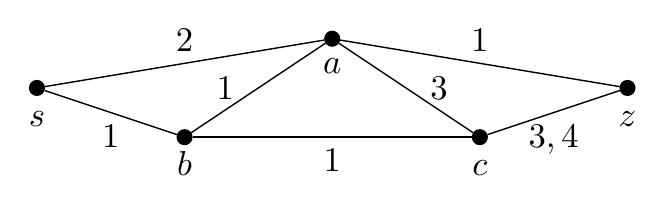}
		\subcaption{\cref{rr:remove-idle} is not applicable}
		\label{fig:rr-after}
	\end{subfigure}
	\caption{
		\cref{fig:rr-before} shows a temporal graph where \cref{rr:remove-idle} is applicable. 
		In particular, layers~$3,4,5,6$ are edgeless. \cref{fig:rr-after} shows the same temporal graph after \cref{rr:remove-idle} was applied exhaustively.
		}
	\label{fig:remove-idle-example}
\end{figure}

	\cref{rr:remove-idle} can be exhaustively applied by iterating over the by time-edges~$(e_i,t_i)$ in the time-edge set~$E$ ordered by ascending labels until the first~$t_1,t_2$ with the given requirement appear. 
	Set~$x_0 := - t_2 + t_1 - 1$.
	Then we iterate further over~$E$ and replace each time-edge~$(e,t)$ with~$(e,t + x_0)$ until the next~$t_1,t_2$ with the given requirement appear.
	Then we set~$x_1 = x_0 - t_2 + t_1 - 1$ and iterate further over~$\TE$ and replace each time-edge~$(e,t)$ with~$(e,t + x_1)$.
	We repeat this procedure until the end of~$\TE$ is reached.
	Since we iterate over~$\TE$ only once, this can be done in~$\ON(|\TE|)$ time.

	\cref{rr:remove-idle-end} can be executed in linear time by iterating over all edges and taking the maximum label as~$t_1$.
Note that the sets~$V$ and~$\TE$ remain untouched by~\cref{rr:remove-idle-end}.
Hence, the application of \cref{rr:remove-idle-end} does not add or remove any \nonstrpath{s,z}. 
\end{proof}
A consequence of \cref{rr-proof:remove-idle} is that the maximum label~$\tau$ can be upper-bounded by the number of time-edges and hence the input size.

\begin{restatable}{lem}{tbound}
	\label{lemma:t-bound}
	Let $\I = (\TG=(V,\TE,\tau),s,z,k)$ be an instance of \sproblem{}.
	There is an algorithm which computes in $\ON(|\TE|)$~time an instance~$\I'=(\TG'=(V,\TE',\tau'),s,z,k)$ of \sproblem{} which is equivalent to~$\I$, where $\tau' \leq |\TE'|$.
\end{restatable}
\begin{proof}
	Let~$\TG = (V,\TE,\tau)$ be a temporal graph, where \cref{rr:remove-idle,rr:remove-idle-end} are not applicable.
	Then for each~$t \in [\tau]$ there is a time-edge~$(e,t) \in \TE$.
	Thus, $\tau \leq \sum_{i=1}^\tau |\set{ (e,t) \in \TE ; t=i}| \leq |\TE|$.
\end{proof}



\section{Hardness Dichotomy Regarding the Number of Layers}
\label{sec:gh}
\appendixsection{sec:gh}
In this section we settle the complexity dichotomy of both \nonstrproblem{} and \strproblem{} regarding the number~$\tau$ of time steps.
We observe that both problems are strongly related to the following \NP{}-complete~\cite{CORLEY1982157,SchieberBK95} problem:
\problemdef{\lenbndproblem{} (LBS)}
{An undirected graph~$G=(V,E)$, distinct vertices~$s,z\in V$, and~$k,\ell\in \N$.}
{Is there a subset~$S\subseteq V\setminus\{s,z\}$ such that~$|S|\leq k$ and there is no $(s,z)$-path in~$G-S$ of length at most~$\ell$?}
\lenbndproblem{} is \NP-complete even if the lower bound~$\ell$ for the path length is five~\cite{Baier2010} and \Wone{}-hard with respect to the postulated separator size~\cite{Golovach201172}.
We obtain the following, improving a result by~\citet{kempe2000connectivity} who showed \NP-completeness of \nonstrproblem{} and \strproblem{} for all $\tau\geq 12$.

\begin{restatable}{theorem}{nphard}
\label{thm:nphard}
\nonstrproblem{} is \NP-complete for every maximum label~$\tau \geq 2$ and \strproblem{} is \NP-complete for every~$\tau \geq 5$.
Moreover, both problems are~\Wone{}-hard when parameterized by the solution size~$k$.
  %
\end{restatable}
%
%
%
\noindent 
We remark that our \NP-hardness reduction for \nonstrproblem{} is inspired by \citet[Theorem~3.9]{Baier2010}.

  \begin{proof}

  To show \NP-completeness of \nonstrproblem{} for $\tau=2$ we present a reduction from the \vertexcover{} problem where, given a graph $G=(V,E)$ and an integer $k$, the task is to determine whether there exists a set $V'\subseteq V$ of size at most $k$ such that $G-V'$ does not contain any edge.
  \paragraph{Construction.}
  Let~$(G=(V,E),k)$ be an instance of \vertexcover{}.
  We say that~$V' \subseteq V$ is a \emph{vertex cover} in~$G$ of size~$k$ if~$|V'| = k$ and~$V'$ is a solution to~$(G=(V,E),k)$.
  We refine the gadget of \citet[Theorem~3.9]{Baier2010} and reduce from \vertexcover{} to \nonstrproblem{}.
  Let~$\I := (G = (V,E),k)$ be a \vertexcover{} instance and~$n := |V|$.
  We construct a \nonstrproblem{} instance~$\I' := (\TG' = (V', \TE', 2),s,z,n+k)$, where 
  $V' := V \cup \set{ s_v,t_v : v \in V } \cup \{s, z \}$
  are the vertices and the time-edges are defined as
  \[
	  \begin{split}
		  \TE' :=\ &\overbrace{\set{ (\{s,s_v\},1),(\{s_v,v\},1),(\{v,z_v\},2),(\{z_v,z\},2),(\{s,v\},2),(\{v,z\},1) : v \in V}}^{\text{vertex-edges}}\ \cup\\
			  &\underbrace{\set{ (\{s_v,z_w\},1),(\{s_w,z_v\},1) : \{v,w\} \in E }}_{\text{edge-edges}}.
	  \end{split}
  \]
  Note that~$|V'| = 3\cdot n  + 2$,~$|\TE'| = 6 \cdot |V'| + 2 \cdot |E|$, and $\I'$ can be computed in polynomial~time.
  For each vertex~$v \in V$ there is a \emph{vertex gadget} which consists of three vertices~$s_v,v,z_v$ and six vertex-edges. 
  In addition, for each edge~$\{v,w\} \in E$ there is an \emph{edge gadget} which consists of two edge-edges~$\{s_v,z_w\}$ and~$\{z_v,s_w\}$.
  See \cref{fig:vc-to-strproblem} for an example.
  \begin{figure}[t!]
	  \centering
	  \includegraphics[width=0.75\textwidth]{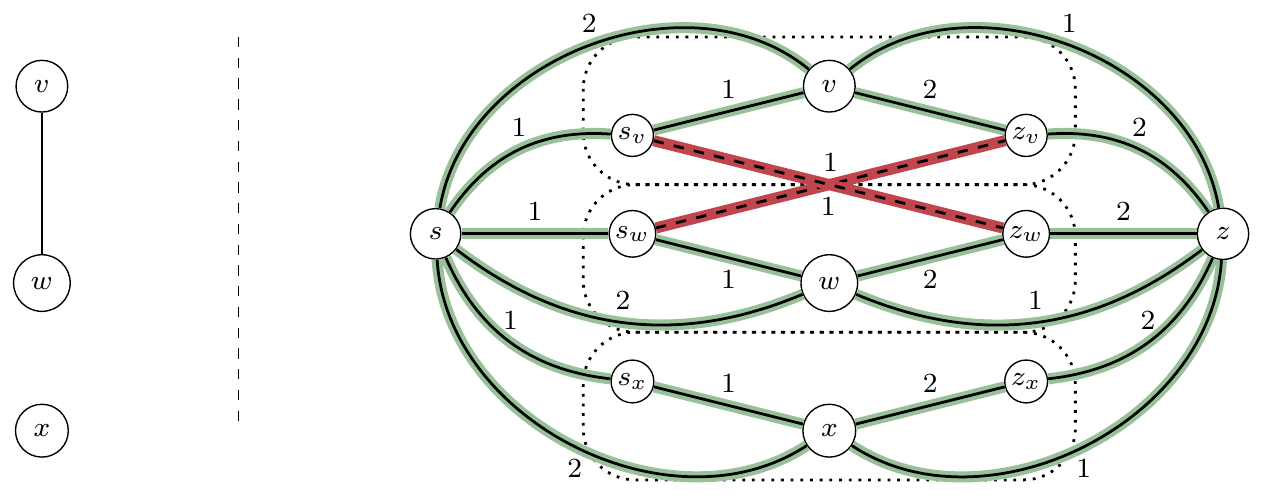}
	  \caption{The \vertexcover{} instance~$(G,1)$ (left) and the corresponding \nonstrproblem{} instance from the reduction of \cref{thm:nphard} (right).
		  The edge-edges are dashed (red), the vertex-edges are solid (green), and the vertex gadgets are in dotted boxes.}
	  \label{fig:vc-to-strproblem}
  \end{figure}
  
\appendixproof{thm:nphard}{\nphard*}
{
	\paragraph{Correctness.}
  We prove that $\I$ is a \yes-instance if and only if $\I'$ is a \yes-instance.

  \implone{}
  Let~$X \subseteq V$ be a vertex cover of size~$k' \leq k$ for~$G$.
  We claim that~$S := (V \setminus X) \cup \set{ s_v,z_v ; v \in X }$ is a \nonstrsep{s,z}.
  There are~$|V| - |X| = n - k'$ vertices not in the vertex cover~$X$ and for each of them there is exactly one vertex in~$S$.
  For each vertex in the vertex cover~$X$ there are two vertices in~$S$.
  Hence,~$|S| = n - k' + 2k' \leq n + k$.
  
  First, we consider the vertex-gadget of a vertex~$v \in V$. 
  Note that in the vertex-gadget of~$v$, there are two distinct \nonstrsep{s,z}s $\{ v \}$ and~$\{ s_v, z_v \}$.
  Hence, every \nonstrpath{s,z} in~$\TG' - S$ contains an edge-edge.
  Second, let~$e = \{ v , w \} \in E$ and let~$P_e$ and~$P_e'$ be the \nonstrpaths{s,z} which contain the edge-edges of edge-gadget of~$e$ such that~$V(P_e) = \{ s, s_v, z_w, z \}$ and~$V(P_e') = \{ s, s_w, z_v,z \}$.
  Since~$X$ is a vertex cover of~$G$ we know that at least one element of~$e$ is in~$X$.
  Thus,~$s_v,z_v \in S$ or~$s_w,z_w \in S$, and hence neither~$P_e$ nor~$P_e'$ exist in~$\hat G - S$.
  It follows that~$S$ is a \nonstrsep{s,z} in~$\TG'$ of size at most~$n+k$, as there are no other \nonstrpath{s,z}s in~$\TG'$.

  \impltwo{}
  Let~$S$ be a \nonstrsep{s,z} in~$\TG'$ of size~$\ell \leq n + k$ and let~$v \in V$.
  Recall that there are two distinct \nonstrsep{s,z}s in the vertex gadget of~$v$, namely~$\{v\}$ and~$\{s_v, z_v\}$, and that all vertices in~$V' \setminus \{s,z\}$ are from a vertex gadget.
  Hence,~$\ell$ is of the form~$\ell=n+k'\leq n+k$.
  We start with a preprocessing to ensure that for vertex gadget only one of these two separators are in~$S$.
  Let~$S_v = S \cap \{ v , s_v, z_w \}$.
  We iterate over~$S_v$ for each~$v \in V$:
  \begin{compactenum}[\bf {Case} 1:]
	  \item If~$S_v = \{ v \}$ or~$S_v = \{ s_v, z_v \}$ then we do nothing.
	  \item If~$S_v = \{ v, s_v, z_v \}$ then we remove~$v$ from~$S$ and decrease~$\ell$ by one.
		  One can observe that all \nonstrpaths{s,z} which are visiting~$v$ are still separated by~$s_v$ or~$w_v$.
	  \item If~$S_v = \{ v, s_v \}$ then we remove~$v$ from~$S$ and add~$z_v$.
		  One can observe that~$S$ is still a \nonstrsep{s,z} of size~$\ell$ in~$\TG'$.
	  \item If~$S_v = \{ v, z_v \}$ then we remove~$v$ from~$S$ and add~$s_v$.
		  One can observe that~$S$ is still a \nonstrsep{s,z} of size~$\ell$ in~$\TG'$.
  \end{compactenum}
  That is a complete case distinction because neither~$\{s_v\}$ nor~$\{z_v\}$ separate all \nonstrpaths{s,z} in the vertex gadget in~$v$.
  Now we construct a vertex cover~$X$ for~$G$ by taking~$v$ into~$X$ if both~$s_v$ and~$z_v$ are in~$S$.
  Since there are~$n$ vertex gadgets in~$\TG'$ each containing either one or two vertices from~$S$, it follows that $|X|=|S| - n = k' \leq k$, 
  
  Assume towards a contradiction that~$X$ is not a vertex cover of~$G$.
  Then there is an edge~$\{v, w\} \in E$ where~$v,w \not \in X$.
  Hence,~$s_v,z_v,s_w,z_w \not \in S$ and~$v,w \in S$.
  This contradicts the fact that~$S$ is a \nonstrsep{s,z} in~$\TG'$, because $P = ((\{s, s_v \},1),(\{ s_v, z_w\},1),$ $(\{z_w, z\},2))$ is a \nonstrpath{s,z} in~$\TG'-S$.
  It follows that~$X$ is a vertex cover of~$G$ of size at most~$k$.
\end{proof}
  \begin{observation}
 \label{obs:lbstostrprob}
 There is a polynomial-time reduction from \textsc{LBS} to \strproblem{} that maps any instance~$(G,s,z,k,\ell)$ of \textsc{LBS} to an instance $(\TG=(G_1,\ldots,G_\ell),s,z,k)$ with~$G_i=G$ for all~$i\in[\ell]$ of \strproblem{}.
  \end{observation}
}

In the remainder of this section we prove that the bound on~$\tau$ is tight in the strict case (for the non-strict case the tightness is obvious). 
This is the first case where we can observe a significant difference between the strict and the non-strict variant of our separation problem.
In order to do so, we have to develop some  tools which we need in subroutines.
In \cref{sec:staticexp}, we introduce a common tool to study reachability in temporal graphs on directed graphs.
This helps us to solve the \textsc{Single-Source Shortest Strict Temporal Paths} efficiently (\cref{lemma:ssssp}).
Note that this might be of independent interest since it improves known algorithms, see \cref{app:adaptation}.
Afterwards, in \cref{sec:str-t-4}, we prove that \strproblem{} can be solved in polynomial time, if the maximum label~$\tau \leq 4$.

\subsection{Strict Static Expansion}\label{sec:staticexp}
A key tool~\cite{berman,kempe2000connectivity,mertzios2013temporal,akrida2017temporal,wu2016efficient} is the time-expanded version of a temporal graph which reduces reachability and other related questions in temporal graphs to similar questions in directed graphs.
Here, we introduce a similar tool for \strpaths{s,z}.
Let~$\TG = (V,\TE,\tau)$ be a temporal graph and let~$V = \{ v_1, \dots, v_{n-2} \}  \cup \{ s, z \}$.
For each $v \in \{v_1,\dots,v_{n-2}\}$, we define the sets~$\phi(v) := \set{ t,t+1 ; t \in[\tau], \exists w : (\{v,w\},t) \in \TE}$
and $\vec{\phi}(v) := \set{ (t,t')\in \phi(v)^2 ; t < t' \land \nexists t'' \in \phi(v) : t< t'' < t' }$.
The \emph{strict static expansion} of~$(\TG,s,z)$ is a directed acyclic graph~$H := (V',A)$ where $V' = \set{ s, z } \cup \set{ u_{t,j} ; j\in[n-2] \land t \in \phi(v_j) }$ and 
$A=A' \cup A_{s} \cup A_{z}\cup A_{\rm col}$, 
$A':=\set{ ( u_{i,j},u_{i+1,j'}), (u_{i,j'}, u_{i+1,j}) ; (\{ v_j,v_{j'}\},i) \in \TE}$,
$A_{s}:=\set{ (s,u_{i+1,j}) ; (\{ s,v_j\},i) \in \TE }$,
$A_{z}:=\set{ (u_{i,j},z) ; (\{ v_j, z\},i) \in \TE }$, and
$A_{\rm col}:=\set{ (u_{t,j},u_{t',j}) ; (t,t') \in \vec{\phi}(v_j)\land j\in[n-2]}$ (referred to as \emph{column-edges} of~$H$).
Observe that each \strpath{s,z} in~$\TG$ has a one-to-one correspondence to some~$(s,z)$-path in~$H$. 
We refer to \cref{fig:static-exp} for an example.
\begin{figure}[t!]
	\includegraphics[width=\textwidth]{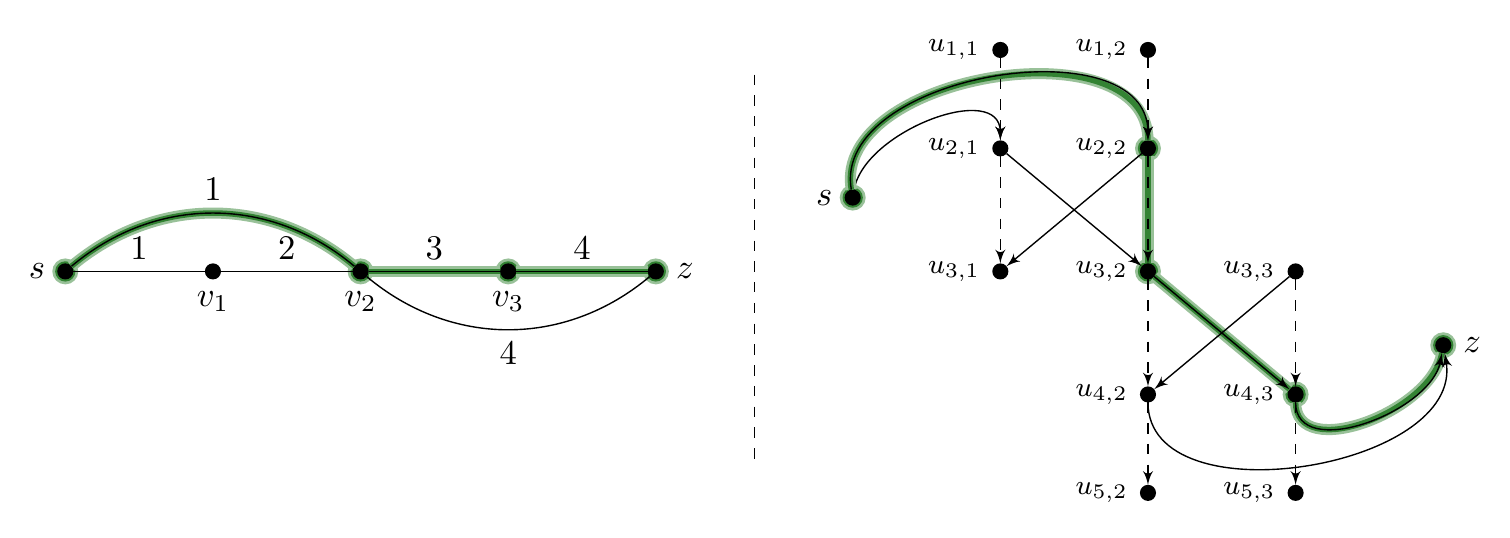}
	\caption{
		A temporal graph $\TG$ (left) and the strict static expansion $H$ for $(\TG,s,z)$ (right).
		One \strpath{s,z} in $\TG$ and its corresponding \npath{s,z} in $H$ are marked (green).
		}
	\label{fig:static-exp}
\end{figure}
\begin{restatable}{lem}{staticexp}
	\label{lemma:const-static-exp}
	Let~$\TG = (V,\TE,\tau)$ be a temporal graph, where~$s,z \in V$ are two distinct vertices.
	The strict static expansion for~$(\TG,s,z)$ can be computed in~$\ON(|\TE|)$ time.
\end{restatable}
\begin{proof}
	Let~$\TG = (V,\TE,\tau)$ be a temporal graph, where~$V = \{ v_1, \dots, v_{n-2} \}  \cup \{ s, z \}$ and~$s,z \in V$ are two distinct vertices.
	Note that $|V| \leq |\TE|$ because $\TG$ is connected.
	We construct the strict static expansion $H= (S,A)$ for~$(\TG,s,z)$ as in four steps follows:
	First, we initiate for each $j\in[n-2]$ an empty linked list.
	Second, we iterate over the set $\TE$ with non-decreasing labels and for each $(\{v_j,v_{j'}\},t)$:
	\begin{compactenum}[(i)]
		\item Add $u_{t,j}, u_{t+1,j}, u_{t,j'}, u_{t+1,j'}$ to $S$ if they do not already exist.
			If we added a vertex $u_{a,b}$ then we push $a$ to the $b$-th linked list.	
		\item Add $(u_{i,j},u_{i+1,j'}), (u_{i,j'}, u_{i+1,j})$ to $A$.
	\end{compactenum}
	This can be done in~$\ON(|\TE|)$ time.
	Observe, that the two consecutive entries in the $i$-th linked list is an entry in $\vec{\phi}(v_i)$.
	Third, we iterate over each linked list in increasing order to add the column-edges to~$A$.
	Note that the sizes of all linked lists sum up to~$\ON(|\TE|)$.
	Last, we add~$s$ and~$z$ to~$S$ as well as the edges in~$A_s$ and~$A_z$ to~$A$.  
	
	Note that the~$|S|$ as well as~$|A|$ can be upper-bounded by~$\ON(|\TE|)$.
	We employed a constant number of procedures each running in~$\ON(|\TE|)$ time. 
	Thus, $H$ can be computed in~$\ON(|\TE|)$ time.
\end{proof}


As a subroutine hidden in several of our algorithms, we need to solve the \textsc{Single-Source Shortest Strict Temporal Paths} problem on temporal graphs:
find shortest strict paths from a source vertex $s$ to all other vertices in the temporal graph.
Herein, we say that a \strpath{s,z} is \emph{shortest} if there is no \strpath{s,z} of length~$\ell' < \ell$.
Indeed, we provide a linear-time algorithm for this.
We believe this to be of independent interest; it improves (with few adaptations to the model; 
for details see \cref{app:adaptation}) 
previous results by \citet{wu2016efficient}, but in contrast to the algorithm of \citet{wu2016efficient} our subroutine cannot be adjusted to the non-strict case. 


\begin{restatable}{proposition}{ssssp}
	\label{lemma:ssssp}
	\textsc{Single-Source Shortest Strict Temporal Paths} is solvable in~$\Theta(|\TE|)$ time. 
\end{restatable}
\appendixproof{lemma:ssssp}{\ssssp*}{
The following proof makes use of a \emph{strict static expansion} of a temporal graph. See \Cref{sec:staticexp} for more details.
\begin{proof}
	By \cref{lemma:const-static-exp},  we compute the strict static expansion~$H = (S,A)$ of~$(\TG,s,z)$ in~$\ON(|\TE|)$ time and define a weight function
	\[
		\omega : A \rightarrow \{0,1\}, (x,y) \mapsto 
		\begin{cases}
			0, \text{if } (x,y) \text{ is a column-edge,}\\
			1, \text{otherwise.}
		\end{cases}
	\]
	
	Observe that~$H$ with~$\omega$ is a weighted directed acyclic graph and that the weight of an \npath{s,z} in~$H$ with~$\omega$ is equal to the length of the corresponding \strpath{s,z} in~$\TG$.
	Hence, we can use an algorithm, which makes use of the topological order of~$S$ on~$H$, to compute for all~$v \in S$ a shortest \npath{s,v} in~$H$ in~$\ON(|S|+|A|)=\ON(|\TE|)$ time (cf.~\citet[Section 24.2]{cormen2009introduction}).

	Now we iterate over $S$ and construct the shortest \strpath{s,w_j} in $\TG$ from the shortest \npath{s,u_{t,j}} in $H$, where $w_j \in V$, $t \in \{1,\dots,\tau\}$ and $u_{t,j} \in S$.
	This can be done in $\ON(|\TE|)$ time because $|S| \in \ON(|\TE|)$.
	Consequently, the overall running time is $\ON(|\TE|)$.
	Since the shortest \strpath{s,z} in $\TG$ can have length $|\TE|$, this algorithm is asymptotically optimal.
\end{proof}}
	\subsection{Adaptation of \cref{lemma:ssssp} for the Model of \citet{wu2016efficient}}\label{app:adaptation}
	\citet{wu2016efficient} considered a model where the temporal graph is directed and a time-edge $e$ has a traversal time $\phi(e) \in \N \cup \{0\}$.
	In the context of strict temporal path $\phi(e)$ is always one.
	They excluded the case where $\phi(e) = 0$, but pointed out that their algorithms can be adjusted to allow~$\phi(e) = 0$.
	However, this is not possible for our algorithm, because then the strict static expansion can contain cycles.
	Hence, we assume that $\phi(e) > 0$ for all directed time-edges $e$.
	
	Let $\TG = (V,\TE,\tau)$ be a directed temporal graph. 
	We denote a directed time-edge from~$v$ to~$w$ in layer $t$ by $((v,w),t) \in \TE$.
	First, we initiate $\tau$ many linked lists.
	Without loss of generality we assume that $\tau \leq |\TE|$, see \cref{lemma:t-bound}.
	Second, we construct a directed temporal graph $\TG' = (V',\TE',\tau)$, where $V=V'$ and $\TE'$ is empty in the beginning.
	Then we iterate over the time-edge set $\TE$ by ascending labels.
	If $((v,w),t) \in \TE$ has $\phi((v,w))=1$ then we add~$((v,w),t)$ to~$\TE'$.
	If $((v,w),t) \in \TE$ has $\phi((v,w))>1$ then we add a new vertex~$x_{(v,w)}$ to~$V'$ and add time-edge~$((v,x_{(v,w)}),t)$ to $\TE'$ and~$((x_{(v,w)},w),t')$ to the~$t'$-th linked list, where~$t'=t+\phi((v,w))-1$.
	We call $((v,x_{(v,w)}),t)$ the \emph{original edge} of $(v,w)$ and $((x_{(v,w)},w),t')$ the \emph{connector edge} of $(v,w)$.
	If we reach a directed time-edge with label~$t$ for the first time, then we add all directed time-edges from the~$t$-th linked list to~$\TE'$.
	Observe that for each \strpath{s,z} in $\TG$ there is a corresponding \strpath{s,z} in $\TG'$, additionally we have that~$\TE'$ is ordered by ascending labels and that~$\TG'$ can be constructed in~$\ON(|\TE|)$ time.

	To construct a strict static expansion for a directed temporal graph $\TG$, we modify the edge set~$A' := \set{ (u_{i,j},u_{i+1,j'}),j') ; ((v_j,v_{j'}),i) \in \TE }$, where~~$((v_j,v_{j'}),i) \in \TE$.
	Finally, we adjust the weight function $\omega$ from the algorithm of \cref{lemma:ssssp} such that $\omega(e)=0$ if $e$ is a column-edges of correspond to a connector-edges, and $\omega(e)=\phi(e)$ otherwise.
	Observe that for a \strpath{s,z} $P=(e_1,\dots,e_\ell)$ of traversal time $\sum_{i=1}^{\ell} \phi(e_i)$ the corresponding \npath{s,z} in the strict static expansion is of weight of the traversal time of~$P$.

\noindent Our algorithm behind~\cref{thm:str-t-4} executes the following steps:
\begin{compactenum}
\item As a preprocessing step, remove unnecessary time-edges and vertices from the graph.
\item Compute an auxiliary graph called \emph{directed path cover graph} of the temporal graph.
\item Compute a separator for the directed path cover graph.
\end{compactenum}
In the following, we explain each of the steps in more detail.

The preprocessing \emph{reduces} the temporal graph such that it has the following properties.
	A temporal graph $\TG=(V,\TE,\tau)$ with two distinct vertices $s,z \in V$ is \emph{reduced} if
	\begin{inparaenum}[(i)]
		\item the underlying graph~$\ug{\TG}$ is connected,
		\item for each time-edge $e \in \TE$ there is a \strpath{s,z} which contains~$e$, and
		\item \label{itm:len-2} there is no \strpath{s,z} of length at most two in~$\TG$.
	\end{inparaenum}
\noindent This preprocessing step can be performed in polynomial time:
\begin{restatable}{lem}{reduceinstance}
	\label{lemma:reduced-instance}
	Let $\I=(\TG=(V,\TE,\tau),s,z,k)$ be an instance of \strproblem{}.
	In $\ON(k\cdot|\TE|)$ time, one can either decide~$\I$ 
	or construct an instance~$\I'=(\TG'=(V',\TE',\tau),s,z,k')$ of \strproblem{} such that $\I'$ is equivalent to $\I$, $\TG'$ is reduced, $|V'| \leq |V|$, $|\TE'| \leq |\TE|$, and $k' \leq k$.
\end{restatable}

\appendixproof{lemma:reduced-instance}{\reduceinstance*}
{
The following proof makes use of a \emph{strict static expansion} of a temporal graph. See \Cref{sec:staticexp} for more details.
\begin{proof}
	First, we remove all time-edges which are not used by a \strpath{s,z}. 
	Let~$\I = (\TG=(V,\TE,\tau)s,z,k)$ be an instance of \strproblem{}.
	We execute the following procedure.
	\begin{compactenum}[(i)]
		\item Construct the strict static expansion~$H = (S,A)$ of~$(\TG,s,z)$.
		\item Perform a breadth-first search in~$H$ from~$s$ and mark all vertices in the search tree as \emph{reachable}.
		      Let~$R(s) \subseteq S$ be the reachable vertices from~$s$.
		\item Construct~$H' := (R(s),A')$, where~$A' := \set{ (w,v) : (v,w) \in A \text{ and } v,w \in R(s) }$.
		      Observe that~$H'$ is the reachable part of~$H$ from~$s$, where all directed arcs change their direction.
	        \item If~$z \notin R(s)$, then our instance~$\I$ is a \yes-instance. 
		\item Perform a breadth-first search from~$z$ in~$H'$ and mark all vertices in the search tree as reachable. 
		      Let~$R(z) \subseteq R(s)$ be the reachable set of vertices from~$z$.
		      In the graph~$H[R(s) \cap R(z)] = H[R(z)]$, all vertices are reachable from~$s$ and from each vertex the vertex~$z$ is reachable.
		\item Output the temporal graph~$\TG' := (V', \TE',\tau)$, where
		     ~$V' := \set{ v_i \in V ; \exists j : x_{j,i} \in R(z) }$ 
		      and~$\TE' := \set{ (\{v_j, v_{j'}\},i) : (x_{(i-1),j}, x_{i,j'}) \in E(H[R(z)]) }$.
	\end{compactenum}
	One can observe that~$\TG'$ is a temporal subgraph of~$\TG$ and that~$\ug{\TG'}$ is connected. 
	Note that all subroutines are computable in $\ON(|\TE|)$ time (see \cref{lemma:const-static-exp}).
	Consequently,~$\TG'$ can be computed in~$\ON(|\TE|)$ time.

	
	We claim that for each time-edge in $\TE'$ there is a \strpath{s,z} in $\TG'$ which contains it.
	Let~$e = (\{v_j,v_{j'}\},i) \in \TE'$ and assume towards a contradiction that there is no \strpath{s,z}~$P$ in~$\TG'$ which contains~$e$.
	From the construction of~$\TG'$ we know that~$(x_{(i-1),j},x_{i,j'}) \in E(H[R(z)])$ or~$(x_{(i-1),j'},x_{i,j}) \in E(H[R(z)])$.
	Hence, there is an \npath{s,z} in~$H$ which contains either~$(x_{(i-1),j},x_{i,j'})$ or~$(x_{(i-1),j'},x_{i,j})$.
	Thus, there is a corresponding \strpath{s,z} in~$\TG$ as well as in~$\TG'$ which contains $(\{v_j,v_{j'}\},i)$.
	This is a contradiction.

	Furthermore, we claim that a \strpath{s,z}~$P$ in~$\TG$ if and only if $P$ is a \strpath{s,z} in~$\TG'$.
	Let $P$ be a \strpath{s,z} in $\TG$.
	Since~$P$ is a \strpath{s,z} in~$\TG$ there is a corresponding \npath{s,z}~$P'$ in~$H$.
	Note that~$P'$ is a witness that all vertices in~$V(P')$ are in~$R(s)$ as well as in~$R(z)$.
	It follows from the construction of~$\TG'$ that~$P$ is also a \strpath{s,z} in~$\TG'$.
	The other direction works analogously.

	Let $P$ be a \strpath{s,z} in $\TG'$ of length two, where $V(P)=\{s,z,v\}$.
	One can observe that each \strsep{s,z} must contain $v$.
	By \cref{lemma:ssssp}, we can compute the shortest \strpath{s,z} $P_s$ in $\TG'$ in $\ON(|\TE|)$ time and check whether $P_s$ is of length two.
	If this is the case, then we remove vertex $v \in V(P_s) \setminus \set{s,z}$ and decrease $k$ by one.
	Note that $\I$ is a \no-instance if we can find $k+1$ vertex-disjoint \strpath{s,z}s of length two and that, this can be done in $\ON(k \cdot |\TE|)$ time.

	If we have not decided yet whether $\I$ is a \yes- or \no-instance, then we construct the \strproblem{} instance $\I'=(\TG'=(V',\TE',\tau),s,z,k')$, where $k'$ is $k$ minus the number of vertex-disjoint \strpath{s,z} of length two we have found.
	
	Finally, observe that $\I'$ is a \yes-instance if and only if $\I$ is a \yes-instance, $\TG'$ is reduced, and that we only removed vertices and time-edges.
	Thus, $|V'| \leq |V|$ and $|\TE'| \leq |\TE|$.

\end{proof}
}

\subsection{Efficient Algorithm for \strproblem{} with Few Layers}
\label{sec:str-t-4}
Now we are all set to show the following result.
\begin{restatable}{theorem}{taufour}
	\label{thm:str-t-4}
	\strproblem{} for maximum label~$\tau\le4$ can be solved in~$\ON(k \cdot |\TE|)$ time, where~$k$ is the solution size.
\end{restatable}

\citet{lovasz1978mengerian} showed that the minimum size of an \nsep{s,z} for paths of length at most four in a graph is equal to the number of vertex-disjoint \npaths{s,z} of length at most four in a graph.
We adjust their idea to strict temporal paths on temporal graphs.
The proof of \citet{lovasz1978mengerian} implicitly relies on the transitivity of connectivity in static graphs.
This does not hold for temporal graphs; hence, we have to put in some extra effort to adapt their result to the temporal case. 
To this end, we define a 
directed auxiliary graph.

\begin{figure}[t!]
  \centering
  \includegraphics[width=1\textwidth]{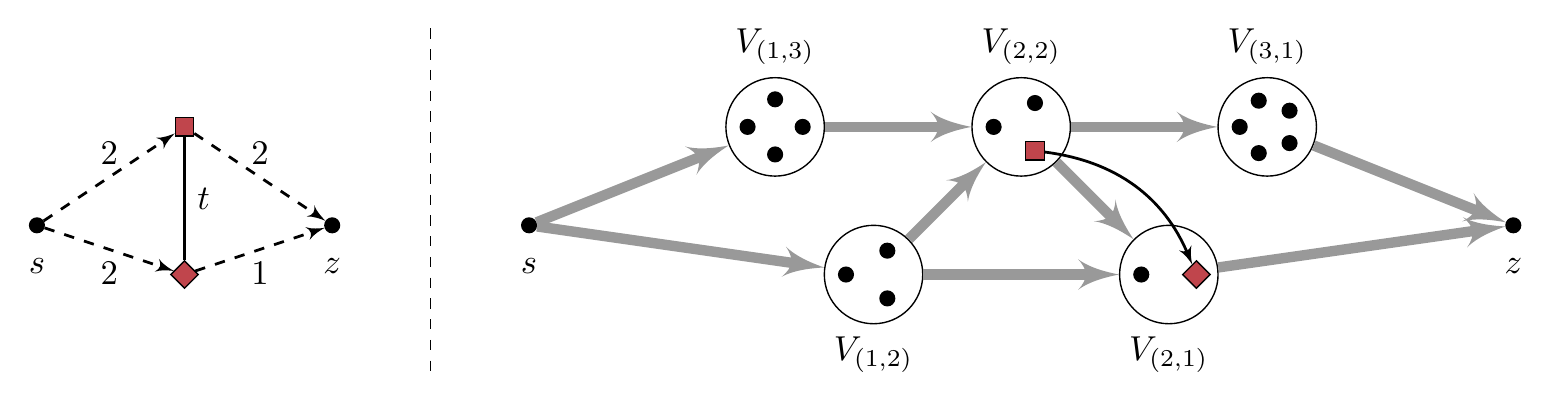}
  \caption{The left side depicts an excerpt of a reduced temporal graph with maximum time-edge label~$\tau=4$. 
      Dashed arcs labeled with a number $x$ indicate a shortest strict temporal path of length~$x$. 
      The right side depicts the directed path cover graph~$D$ from~$s$ to~$z$ of the reduced temporal graph. A gray arc from vertex set~$V_{(i,j)}$ to vertex set~$V_{(i',j')}$ denotes that for two vertices~$v \in V_{(i,j)}$ and~$w \in V_{(i',j')}$ there can be an arc from~$v$ to~$w$ in~$D$. Take as an example the square-shaped vertex in $V_{(2, 2)}$ and the diamond-shaped vertex in~$V_{(2,1)}$.}
  \label{fig:dpcg}
\end{figure}
\begin{defin}[Directed Path Cover Graph]
	\label{def:dpcg}
	Let~$\TG=(V,\TE,\tau=4)$ be a reduced temporal graph with~$s,z \in V$.
	The \emph{directed path cover graph} from~$s$ to~$z$ of~$\TG$  is a directed graph~$D = (V,\vec{E})$ such that~$(v,w) \in \vec{E}$ if and only if
	\begin{inparaenum}[(i)]%
		\item $v,w \in V$,
		\item $(\{v,w\},t) \in \TE$ for some~$t\in[\tau]$, and
		\item $v \in V_{(i,j)}$ and~$w \in V_{(i',j')}$ such that~$i < i'$, ~$v \in V_{(2,2)}$ and~$w \in V_{(2,1)}$,~$v=s$ and~$w\in V_{(1,j)}$, or~$w=z$ and~$v\in V_{(i,1)}$ for some~$i,j\in\{2,3\}$.
	\end{inparaenum}
	Herein, a vertex~$x \in V$ is in the set~$V_{(i,j)}$ if the shortest \strpath{s,x} is of length~$i$ and the shortest \strpath{x,z} is of length~$j$.
\end{defin}
\noindent
\cref{fig:dpcg} depicts a generic directed path cover graph of a reduced temporal graph with~$\tau=4$. 
Note that due to 
the definition
of reduced temporal graphs, one can prove that the set $V_{(1, 1)}$ is always empty, and hence not depicted in~\cref{fig:dpcg}. 
This is a crucial property that allows us to prove the following. 
	
\begin{restatable}{lem}{mengerian}
	\label{lemma:tau-4-mengerian}
	Let~$\TG=(V,\TE,\tau=4)$ be a reduced temporal graph with~$s,z \in V$.
	Then the directed path cover graph~$D$ from $s$ to $z$ of $\TG$ can be computed in~$\ON(|\TE|)$ time and~$S \subseteq V\setminus \{s,z\} $ is a \strsep{s,z} in~$\TG$ if and only if~$S$ is an~\nsep{s,z} in~$D$.
\end{restatable}
\appendixproof{lemma:tau-4-mengerian}{\mengerian*}{
We introduce some notation we will use in the following proof.
The \emph{departure time} (\emph{arrival time}) of a \nonstrpath{s,z}~$P = ((e_1,t_1),\dots,(e_\ell,t_\ell))$ is~$t_1$ ($t_{\ell}$), the \emph{traversal time} of~$P$ is~$t_{\ell} - t_1$, and the length of~$P$ is~$\ell$.
\begin{proof}

	Let~$\TG=(V,\TE,\tau=4)$ be a reduced temporal graph with~$s,z \in V$.
	In order to compute the directed path cover graph $D=(V,\vec{E})$ from $s$ to $z$ of $\TG$, we need to know the shortest \strpath{s,v} and the shortest \strpath{v,z} in $\TG$ for all $v \in V$.
	For~$s$ this can be done in $\ON(|\TE|)$ time by \cref{lemma:ssssp}.
	For~$z$, consider the bijection~$\lambda:[\tau]\to[\tau]$ with~$\lambda(t)=\tau+1-t$ and the helping graph~$\overleftarrow{\TG}:=(V,\overleftarrow{\TE},\tau)$ of~$\TG$ where~$\overleftarrow{\TE}:=\{(\{v,w\},\lambda(t))\mid (\{v,w\},t)\in\TE\}$.
	Note that~$\overleftarrow{\TG}$ can be constructed in~$O(|\TE|)$ time.
	Moreover, it holds true that~$P$ is a shortest \strpath{v,z} in~$\TG$ if and only if~$\overleftarrow{P}$ is a shortest \strpath{z,v} in $\overleftarrow{\TG}$, where~$\overleftarrow{P}$ is~$P$ where each temporal edge~$(\{x,y\},t)$ in~$P$ is replaced by~$(\{x,y\},\lambda(t))$.
	If~$P$ is not a shortest~\strpath{v,z} in~$\TG$, then there is a shorter \strpath{v,z}~$P'$ in~$\TG$.
	Then~$\overleftarrow{P'}$ is a \strpath{z,v} in $\overleftarrow{\TG}$ shorter than~$\overleftarrow{P}$.
	The other direction is proven analogously.
	It follows that computing the shortest \strpath{v,z} in~$\TG$ for all~$v \in V$ can be done in $\ON(|\TE|)$ time by employing \cref{lemma:ssssp} in~$\overleftarrow{\TE}$.
	
	Now we iterate over the time-edge set $\TE$ and check for each $(\{v,w\},t)$ whether there is an arc between $v$ and $w$ in $\vec{E}$.
	Hence, the directed path cover graph $D=(V,\vec{E})$ from $s$ to $z$ of $G$ can be computed in $\ON(|\TE|)$.	 

	We say that a \strpath{s,z}~$P$ of length~$n$ in~$\TG$ is \emph{chordless} if~$\TG[V(P)]$ does not contain a \strpath{s,z}~$P'$ of length~$n-1$.
	If such a~$P'$ exists, we call~$P'$ a \emph{chord} of~$P$.
	Obviously, for a set~$X$ of vertices, if~$\TG-X$ does not contain any chordless \strpath{s,z}, then~$\TG-X$ does not contain any \strpath{s,z}.

	Now, we discuss the appearance of the directed path cover graph~$D$.
	Assume towards a contradiction that there is a vertex~$v \in V_{(1,1)}$.
	Since $\TG$ is reduced, there is no \strpaths{s,z} of length two.
	Thus, if there is a vertex~$v \in V_{(1,1)}$, then there must be time-edges~$(\{s,v\},t_1)$ and~$(\{v,z\},t_2)$ such that~$t_2 \leq t_1$.
	We distinguish two cases.
	\begin{description}
		\item[Case 1:] Let~$t_2 \leq 2 \leq t_1$.
			Since $\TG$ is reduced, we know that there is a \strpath{s,z}~$P$ which contains~$(\{v,z\},t_2)$.
			The \strpath{s,z}~$P$ must be of length at least three, because $\TG$ is reduced
			Observe that the length of a \strpath{s,z} is a lower bound for the arrival time.
			Hence, the arrival time of~$P$ is at least three.
			This is a contradiction because~$(\{v,z\},t_2)$ is the last time-edge of~$P$ and, therefore,~$t_2 \leq 2$ is equal to the arrival time.
			Consequently,~$v \notin V_{(1,1)}$.
		\item[Case 2:] Let~$3 \leq t_2 \leq t_1$.
			Since $\TG$ is reduced, we know that there is a \strpath{s,z}~$P$ which contains~$(\{s,v\},t_1)$.
			The \strpath{s,z}~$P$ must be of length at least three because $\TG$ is reduced.
			Observe that~$(\{s,v\},t_1)$ is the first time-edge of~$P$.
			Hence, the departure time of~$P$ is at least three.
			Since,~$P$ is of length at least three, the arrival time of~$P$ is at least five.
			This is a contradiction because the maximum label is four.
			Consequently,~$v \notin V_{(1,1)}$.
	\end{description}
	Observe that the sets~$V_{(1,3)},V_{(2,2)},V_{(3,1)},V_{(1,2)},V_{(2,1)}$ are a partition of~$V\setminus\{s,t\}$, because all \strpaths{s,z} are of length at most four and~$V_{(1,1)} = \emptyset$. 
	In \cref{fig:dpcg}, we can see the appearance of the directed path cover graph~$D$ from~$s$ to~$z$ of~$\TG$.

	We claim that for each chordless \strpath{s,z}~$P$ in~$\TG$, there is an \npath{s,z}~$P_D$ such that~$V(P)=V(P_D)$.
	Let~$P$ be a chordless \strpath{s,z} in~$\TG$.
	Since $\TG$ is reduced,~$P$ is of length three or four.
	\begin{description}
		\item[Case 1:] Let the length of~$P$ be three and~$V(P)= \{s,v_1,v_2,z\}$ such that~$v_1$ is visited before~$v_2$.
			Since there is no \strpath{s,z} of length two,~$P$ forces~$v_1$ into~$V_{(1,2)}$ and~$v_2 \in V_{(2,1)}$.
			It is easy to see that~$v_1 \in V_{(1,2)}$, just by the existence of~$P$.
			For~$v_2 \in V_{(2,1)}$, we need that there is no \strpath{s,v_2} of length one.
			
			Assume towards a contradiction that there would be a \strpath{s,v_2} of length one or in other words there is a time-edge~$(\{s,v_2\},t) \in \TE$.
			Then~$v_2\in V_{(1,1)}$, which contradicts the fact that~$V_{(1,1)}=\emptyset$, and hence~$v_2 \in V_{(2,1)}$.
			

			From \cref{def:dpcg}, we know that there is an \npath{s,z}~$P_D$ such that~$V(P) = V(P_D)$.
		\item[Case 2:] Let the length of~$P$ be four.
			Thus, it looks as follows. 
			\begin{equation*}
				P = ((\{s,v_1\},1),(\{v_1,v_2\},2),(\{v_2,v_3\},3),(\{v_3,z\},4)).
			\end{equation*}
			It immediately follows that the shortest \strpath{s,v_1} is of length one.
			Furthermore, the shortest \strpath{s,v_2} is of length two, otherwise~$P$ would not be chordless.

			Now, we claim that~$v_3 \in V_{(2,1)} \cup V_{(3,1)}$.
			Assume towards a contradiction that~$v_3 \notin V_{(2,1)} \cup V_{(3,1)}$.
			Observe that~$v_3 \notin V_{(2,2)}$ because~$(\{v_3,z\},4)$ is a shortest \strpath{v_3,z} of length one.
			Thus, there is a time-edge~$(\{s,v_3\},t) \in \TE$. 
			This time-edge must be part of a \strpath{s,z}, because $\TG$ is reduced.
			Since~$v_3 \neq z$, we have~$t\leq 3$.
			Note that~$((\{s,v_3\},t\leq3),(\{v_3,z\},4))$ would be a \strpath{s,z} of length two.
			Hence,~$(\{s,v_3\},t) \notin \TE$ and~$v_3 \in V_{(2,1)} \cup V_{(3,1)}$.

			From \cref{def:dpcg} we know that there is an \npath{s,z}~$P_D$ such that~$V(P) = V(P_D)$.
	\end{description}

	Let~$S \subseteq (V \setminus \{s,z\})$ be an \nsep{s,z} in~$D$.
	Until now, we know that~$S$ is also a \strsep{s,z} in~$\TG$, because for each chordless \strpath{s,z} in~$\TG$ there is an \npath{s,z} in~$D$.
	Let~$L \subseteq (V \setminus \{s,z\})$ be a \strsep{s,z} in~$\TG$.
	It remains to be shown that~$L$ is an \nsep{s,z} in~$D$.

	Assume towards a contradiction that~$L$ is not an \nsep{s,z} in~$D$.
	Thus, there is an \npath{s,z}~$P_D$ in~$D-L$.
	The length of~$P_D$ is either three or four, see \cref{fig:dpcg}.
	\begin{description}
		\item[Case 1:] Let the length of~$P_D$ be three and~$V(P_D) = \{s, v_1, v_2, z\}$ such that~$v_1$ is visited before~$v_2$.
			Thus,~$v_1 \in V_{(1,2)}$ and~$v_2 \in V_{(2,1)}$, see \cref{fig:dpcg}.
			From \cref{def:dpcg}, we know that there are time-edges~$(\{s,v_1\},t_1),(\{v_1,v_2\},t_2),(\{v_2,z\},t_3) \in \TE$.
			Note that each time-edge in~$\TG$ must be part of a \strpath{s,z} and that there are no \strpaths{s,z} of length two.
			Hence,~$t_1 \in \{1,2 \}$,~$t_2 \in \{2,3\}$, and~$t_3 \in \{3,4\}$.
%
%
			We distinguish the cases~$t_2=2$ and~$t_2=3$.

			\begin{description}
				\item[Case 1.1:] Let~$t_2=2$.
					The time-edge~$(\{v_1,v_2\},t_2=2)$ must be part of a \strpath{s,z}, and therefore we have either~$(\{s,v_2\},1) \in \TE$ or~$(\{s,v_1\},1) \in \TE$.
					Time-edge~$(\{s,v_2\},1)$ cannot exist since~$v_2 \in V_{(2,1)}$.
					Hence,~$(\{s,v_1\},1) \in \TE$,~$t_1=1$, and~$((\{s,v_1\},t_1),(\{v_1,v_2\},t_2),(\{v_2,z\},t_3))$ is a \strpath{s,z} in~$\TG$.
					Moreover,~$\{v_1,v_2\} \cap L \neq \emptyset$.
					This contradicts the fact that~$L$ is \strsep{s,z} in~$\TG$.
				\item[Case 1.2:] Let~$t_2=3$.
					There is a \strpath{s,z} which contains the time-edge~$(\{v_1,v_2\},t_2=3)$.
					Since~$t_3\geq 3$, there is either a~$(\{v_1,z\},4) \in \TE$ or~$(\{v_2,z\},4) \in \TE$.
					The time-edge~$(\{v_1,z\},4)$ cannot exist because~$v_1 \in V_{(1,2)}$.
					Hence~$(\{v_2,z\},4) \in \TE$,~$t_3 = 4$, and~$((\{s,v_1\},t_1),(\{v_1,v_2\},t_2),(\{v_2,z\},t_3))$ is a \strpath{s,z} in~$\TG$.
					Moreover,~$\{v_1,v_2\} \cap L \neq \emptyset$.
					This contradicts the fact that~$L$ is \strsep{s,z} in~$\TG$.

			\end{description}

		\item[Case 2:] Let the length of~$P_D$ be four, and~$V(P_D) = \{s,v_1,v_2,v_3,z\}$ such that~$v_1$ is visited before~$v_2$ and~$v_2$ is visited before~$v_3$.
			Considering \cref{fig:dpcg}, one can observe that~$v_1 \in V_{(1,3)} \cup V_{(1,2)}$ and~$v_2 \in V_{(2,2)}$.
			From \cref{def:dpcg}, we know that there are time-edges~$(\{s,v_1\},t_1),(\{v_1,v_2\},t_2),$ $(\{v_2,v_3\},t_3),(\{v_3,z\},t_4) \in \TE$.
			There is a \strpath{s,z}~$P$ containing $(\{v_1,v_2\},t_2)$ because $\TG$ is reduced. 
			Since~$v_2 \in V_{(2,2)}$, we know that~$P$ is of length four.

			Now assume towards a contradiction that~$P$ visits~$v_2$ before~$v_1$.
			Hence,~$(\{v_1,v_2\},t_2)$ is the third time-edge in~$P$.
			The first two time-edges of~$P$ are a \strpath{s,v_2} of length two and arrival time at least two.
			Thus,~$t_2 \geq 3$.
			Observe that there is no \strpath{v_1,z} of length one and hence there is no time-edge~$(\{v_1,z\},4)$, because~$v_1 \in V_{(1,2)} \cup V_{(1,3)}$.
			This contradicts~$P$ visiting~$v_2$ before~$v_1$. 
			Consequently, the \strpath{s,z}~$P$ visits~$v_1$ before~$v_2$ and ~$(\{v_1,v_2\},t_2)$ is the second time-edge in~$P$.
			Hence,~$t_2=2$ and there is a time-edge~$(\{s,v_1\},1)$ in~$P$.
			This implies~$t_1=1$.

			The vertex~$v_3 \in V_{(2,1)} \cup V_{(3,1)}$, see \cref{fig:dpcg}.
			There is a \strpath{s,z} which contains~$(\{v_2,v_3\},t_3)$, because $\TG$ is reduced. 
			Since~$v_3 \in V_{(2,1)} \cup V_{(3,1)}$ and~$v_2 \in V_{(2,2)}$, all \strpaths{s,v_2} and \strpaths{s,v_3} have length at least two and thus also arrival time at least two.
			Hence,~$t_3 \geq 3$.
			Because~$v_2 \neq z \neq v_3$ and there is a \strpath{s,z} which contains~$(\{v_2,v_3\},t_3)$, we have~$t_3 < 4 \impl t_3=3$ and there is either a time-edge~$(\{v_2,z\},4) \in \TE$ or~$(\{v_3,z\},4) \in \TE$.
			The time-edge~$(\{v_2,z\},4)$ cannot exist, otherwise there would be a \strpath{v_2,z} of length one, but~$v_2 \in V_{(2,2)}$.
			This implies that~$(\{v_3,z\},4) \in \TE$ and~$t_4=4$.
			Hence, the time-edge sequence~$((\{s,v_1\},t_1),(\{v_1,v_2\},t_2),$ $(\{v_2,v_3\},t_3),(\{v_3,z\},t_4))$ forms a \strpath{s,z} and~$\{v_1,v_2,v_3 \} \cap L \neq \emptyset$.
			This contradicts the fact that~$L$ is \strsep{s,z} in~$\TG$.
	\end{description}

	In summary,~$P_D$ cannot exist in~$D-L$.
	Consequently,~$L$ is an \nsep{s,z} in~$D$.
\end{proof}
}
\noindent \Cref{fig:non-merg-tau-5} shows that if $\tau=5$, then we can construct a reduced temporal graph where the set $V_{(1, 1)}$ is not empty. 
This indicates why our algorithm fails for $\tau=5$.
\begin{figure}[t!]
	\centering
	\includegraphics[width=0.5\textwidth]{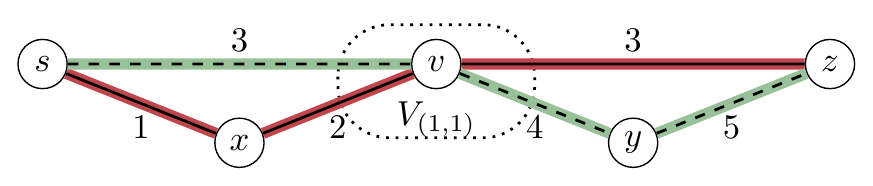}
	\caption{
		A reduced temporal graph with maximum label~$\tau=5$ where the vertex set~$V_{(1,1)}$ of the directed path cover graph is not empty. The solid (red) and dashed (green) edges are strict temporal paths and show that edges~$(\{s, v\},3)$ and~$(\{v, z\},3)$ are not removed when the graph is reduced.
		Furthermore, $v$ is not removed since~$((\{s, v\},3),(\{v, z\},3))$ is not a strict temporal path.
	}
	\label{fig:non-merg-tau-5}
\end{figure}

Finally, with \cref{lemma:reduced-instance,lemma:tau-4-mengerian} we can prove \cref{thm:str-t-4}.
  \begin{proof}[Proof of \cref{thm:str-t-4}]
	  Let~$\I := (\TG=(V,\TE,\tau=4),s,z,k)$ be an instance of \strproblem{}. 
	  First, apply \cref{lemma:reduced-instance} in~$\ON(k\cdot |\TE|)$~time to either decide~$\I$ or to obtain an instance~$\I'=(\TG'=(V',\TE',\tau),s,z,k')$ of \strproblem{}. 
	  In the second case, compute the directed path cover graph~$D$ of~$\TG'$ from~$s$ to~$z$ in~$\ON(|\TE'|)$ time (by \cref{lemma:tau-4-mengerian}).
	  Next, check whether~$D$ has an \nsep{s,z} of size at most~$k'$ in~$\ON(k' \cdot |\TE'|)$ time by a folklore result~\cite{ford1956maximal}.
	  By \cref{lemma:tau-4-mengerian},~$D$ has an \nsep{s,z} of size~$k'$ if and only if $\TG'$ has a \strsep{s,z} of size~$k'$.
    	  Since by~\cref{lemma:reduced-instance} we have that~$\TG'$ is reduced,~$|V'|\leq |V|$, $|\TE'|\leq |\TE|$, and~$k'\leq k$, the overall running time is~$\ON(k\cdot|\TE|)$.
  \end{proof}


\section{On Temporal Graphs with Planar Underlying Graph}
\label{sec:nonstr}
\appendixsection{sec:nonstr}


In this section, we study our problems on \emph{planar temporal graphs}, that is, temporal graphs that have a planar underlying graph.
We show that both \nonstrproblem{} and \strproblem{} remain \NP-complete on planar temporal graphs.
On the positive side, we show that on planar temporal graphs with  a constant number of layers, \strproblem{} can be solved in~$\ON(|\TE|\cdot\log|\TE|)$ time.

In order to prove our hardness results, we first prove \NP-hardness for \lenbndproblem{} on planar graphs---a result which we consider to be of independent interest; note that \NP-completeness on planar graphs was only known for the edge-deletion variant of \lenbndproblem{} on undirected graphs~\cite{FHNN18} and weighted directed graphs~\cite{PanS16}.%
\begin{restatable}
	{theorem}{npplanar}
	\label{lemma:np-hard-planar}
	\lenbndproblem{} on planar graphs is \NP-hard. 
\end{restatable}
%

\begin{proof}
	We give a many-one reduction from the \NP-complete~\cite{FHNN18} edge-weighted variant of \lenbndcutproblem{}, referred to as \textsc{Planar \lenbndcutproblem{}}, where the input graph~$G=(V,E)$ is planar, has edge costs~$c:E\to\{1,k+1\}$, has maximum degree~$\Delta=6$, the degree of~$s$ and~$z$ is three, and $s$ and $z$ are incident to the outer face. 
	Since the maximum degree is constant, one can replace a vertex with a planar grid-like gadget.
	
	Let~$\I := (G=(V,E,c),s,z,\ell,k)$ be an instance of \textsc{Planar \lenbndcutproblem{}}, and we assume~$k$ to be even\footnote{If~$k$ is odd, since~$s$ and~$z$ are incident to the outer face, then we can add a path of length~$\ell-1$ with endpoints~$s$ and $z$ and set the budget for edge deletions to~$k+1$.}.
	We construct an instance~$\I' := (G',s',z',\ell',k)$ of \lenbndproblem{} as follows (refer to~\cref{fig:np-planar} for an illustration).

	\smallskip\noindent\emph{Construction.}
	For each vertex~$v \in V$, we introduce a \emph{vertex-gadget}~$G_v$ which is a grid %
	of size~$(2k+2) \times (2k+2)$, that is, a graph with vertex set~$~\{u^v_{i,j}\mid i,j\in[2k+2]\}$ and edge set~$\{\{u^v_{i,j},u^v_{i',j'}\}\mid|i-i'|+|j-j'|=1\}$.
	There are six pairwise disjoint subsets~$C_v^1,\dots,C_v^6 \subseteq V(G_v)$ of size~$k+1$ that we refer to as \emph{connector sets}.
	As we fix an orientation of~$G_v$ such that $u^v_{1,1}$ is in the top-left, there are two connector sets on the top of~$G_v$, two on the bottom of~$G_v$, one on the left of~$G_v$, and one on the right of~$G_v$. 
	Formally, $C_v^1=\{u^v_{1,k+2},\ldots,u^v_{1,2k+2}\}$, $C_v^2=\{u^v_{k/2,2k+2},\ldots,u^v_{3k/2,2k+2}\}$, $C_v^3=\{u^v_{2k+2,k+2},\ldots,u^v_{2k+2,2k+2}\}$, $C_v^4=\{u^v_{2k+2,1},\ldots,u^v_{2k+2,k+1}\}$, $C_v^5=\{u^v_{k/2,1},\ldots,u^v_{3k/2,1}\}$, and $C_v^6=\{u^v_{1,1},\ldots,u^v_{1,k+1}\}$.

	\begin{figure}
	  \includegraphics[width=\textwidth]{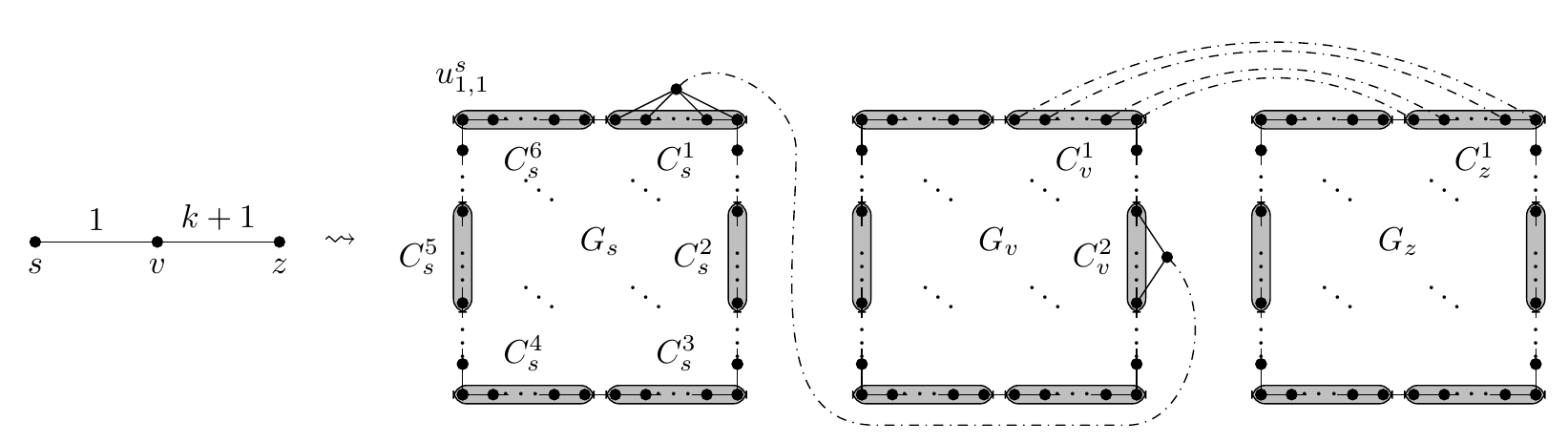}
	  \caption{A simple planar graph~$G$ (left) with edge costs (above edges) and the obtained graph~$G'$ in the reduction from~\cref{lemma:np-hard-planar}. 
	      The connector sets are highlighted in gray. 
	      The edge-gadgets are indicated by dash-dotted lines.}
	      \label{fig:np-planar}
	\end{figure}

	Note that all \npaths{x,y} are of length at most~$k' := (2k+2)^2-1$, for all~$x,y \in V(G_v)$, because there are only~$(2k+2)^2$ vertices in~$V(G_v)$.

	Let~$\phi(G)$ be a plane embedding of~$G$.
	We say that an edge~$e$ incident with vertex~$v \in V$ is \emph{at position~$i$ on~$v$} if~$e$ is the $i$th edge incident with~$v$ when counted clockwise with respect to~$\phi(G)$.
	For each edge~$e=\{v,w\}$, we introduce an \emph{edge-gadget}~$G_{e}$ that differs on the weight of~$e$, as follows.
	Let~$e$ be at position~$i \in \{1, \dots, 6\}$ on~$v$ and at position~$j \in \{1, \dots, 6\}$ on $w$.
	If~$c(e)=1$, then~$G_{e}$ is constructed as follows.
	Add a path consisting of~$(\ell +1)\cdot k' - 1$ vertices and connect one endpoint with each vertex in~$C_v^i$ by an edge and connect the other endpoint with each vertex in~$C_w^j$ by an edge. 
	If~$c(e)=k+1$, then~$G_{e}$ is constructed as follows.
	We introduce a planar matching between the vertices in~$C_v^i$ and~$C_w^j$.
	That is, for instance, we connect vertex~$u^v_{1,k+2+p}$ with vertex~$u^w_{1,2k+2-p}$ for each~$p\in\{0,\ldots,k\}$, if $i=j=1$, or we connect vertex~$u^v_{1,1+p}$ with vertex~$u^w_{2k+2,3k/2-p}$ for each~$p\in\{0,\ldots,k\}$, if $i=6$ and $j=2$ (we omit the remaining cases).
	Then, replace each edge by a path of length at least~$(\ell +1)\cdot k' +1$ where its endpoints are identified with the endpoints of the replaced edge.
	Hence, a path between two vertex-gadgets has length at least~$(\ell +1)\cdot k' + 1$.
	
	Next, we choose connector sets~$C_s^{i'}$ and~$C_z^{j'}$ such that no vertex~$v \in C_s^{i'} \cup C_z^{j'}$ is adjacent to a vertex from an edge-gadget.
	Such~$i'$ and~$j'$ always exist because the degrees of~$s$ and~$z$ are both three.
	Now, we add two special vertices~$s'$ and~$z'$ and edges between~$s'$ and each vertex in~$C_s^{i'}$, as well as between~$z'$ and each vertex in~$C_z^{j'}$.

	Finally, we set 
	$
	    \ell' := 2 + (\ell+1) \cdot k' + \ell \left( (\ell+1)\cdot k' +1 \right).
	$
	Note that 
	$G'$ can be computed in polynomial~time.
	Moreover, one can observe that~$G'$ is planar by obtaining an embedding from~$\phi$.
	This concludes the description of the construction.
   
	\smallskip\noindent\emph{Correctness.}
	We claim that~$\I$ is a \yes-instance if and only if~$\I'$ is a \yes-instance.

	\implone{}
	Let~$\I$ be a \yes-instance.
	Thus, there is a solution~$C \subset E$ with~$c(C)\leq k$ such that there is no \npath{s,z} of length at most~$\ell$ in~$G - C$.
	We construct a set~$S \subset V(G')$ of size at most~$k$ by taking for each~$\{v,w\} \in C$ one arbitrary vertex from the edge-gadget~$G_{\{v,w\}}$ into~$S$.
	Note that since~$c(C)\leq k$, each edge in~$C$ is of cost one.
	
	Assume towards a contradiction that there is a shortest \npath{s',z'}~$P'$ of length at most~$\ell'$ in~$G' - S$.
	Since a path between two vertex-gadgets has length at least~$(\ell +1)\cdot k' + 1$, we know that~$P'$ goes through at most~$\ell$ edge-gadgets.
	Otherwise~$P'$ would be of length at least 
	$
	    2 + (\ell + 1) \cdot \left[(\ell + 1) \cdot k' + 1 \right]
	    = 2 + (\ell + 1)\cdot k' + \ell \cdot \left[(\ell + 1) \cdot k' + 1 \right] + 1 = \ell'+1.		
	$
	Now, we reconstruct an \npath{s,z}~$P$ in~$G$ corresponding to~$P'$ by taking an edge~$e \in E$ into~$P$ if~$P'$ goes through the edge-gadget~$G_e$.
	Hence, the length of~$P$ is at most~$\ell$.
	This contradicts that there is no \npath{s,z} of length at most~$\ell$ in~$G - C$.
	Consequently, there is no \npath{s',z'} of length at most~$\ell'$ in~$G'-S$ and~$\I'$ is a \yes-instance.

	\impltwo{}
	Let~$\I'$ be a \yes-instance.
	Thus, there is a solution~$S \subseteq V(G')$ of minimum size (at most~$k$) such that there is no \npath{s',z'} of length at most~$\ell'$ in~$G' - S$.
	Since~$S$ is of minimum size, it follows from the following claim that~$V(G_v)\cap S=\emptyset$ for all~$v\in V$.
	\begin{claim}
	  \label{claim:planar-np}
	  Let~$G_v$ be a vertex-gadget and~$i, j \in \{1,\dots, 6 \}$ with~$i \neq j$.
	  Then, for each vertex set~$S \subseteq V(G_v)$ of size at most~$k$ it holds that there are~$v_1 \in C_v^i \setminus S$ and~$v_2 \in C_v^j \setminus S$ such that there is a \npath{v_1,v_2} of length at most~$k'$ in~$G_v-S$.
	\end{claim}
	\begin{proof}[Proof of \cref{claim:planar-np}]
	  	  Let~$C_v^i, C_v^j$ two connector sets of a vertex-gadget~$G_v$, where~$i,j \in \{1,\dots, 6 \}$ and~$i \not = j$.
	  We add vertices~$a$ and $b$ and edges~$\{a,a'\}$ and~$\{b,b'\}$ to~$G_v$, where~$a' \in C_v^i$ and~$b' \in C_v^j$.
	  There are~$6 \choose 2$ different cases in which~$i \not = j$.
	  It is not difficult to see that in each case there are~$k+1$ vertex-disjoint \npaths{a,b}.
	  The claim then follows by Menger's Theorem~\cite{Menger1927}.
	\end{proof}
	Note that by minimality of~$S$, it holds that~$V(G_{e}) \cap S=\emptyset$ for all~$e\in E$ with~$c(e)=k+1$.
	We construct an edge set~$C \subseteq E$ of cost at most~$k$ by taking~$\{v,w\} \in E$ into~$C$ if there is a~$y \in V(G_{\{v,w\}}) \cap S$.

	Assume towards a contradiction that there is a shortest \npath{s,z}~$P$ of length at most~$\ell$ in~$G - C$.
	We reconstruct an \npath{s',z'}~$P'$ in~$G'$ which corresponds to~$P$ as follows.
	First, we take an edge~$\{s', v \} \in E(G')$ such that~$v \in C_s^{i'} \setminus S$.
	Such a~$v$ always exists, because~$|C_s^{i'}| = k+1$ and~$|S| \leq k$.
	Let~$\{s,w\} \in E$ be the first edge of~$P$ and at position~$i$ on~$w$.
	Then we add a \npath{v,v'}~$P_s$ in~$G_s - S$, such that~$v' \in C_s^i \setminus S$.
	Due to~\cref{claim:planar-np}, such a \npath{v,v'}~$P_s$ always exists in~$G_s - S$ and is of length at most~$k'$.

	We take an edge-gadget~$G_e$ into~$P'$ if~$e$ is in~$P$.
	Recall, that an edge-gadget is a path of length~$(\ell+1) \cdot k'+1$.
	Due to~\cref{claim:planar-np}, we can connect the edge-gadgets~$G_{\{v_1,v_2\}},G_{\{v_2,v_3\}}$ of two consecutive edges~$\{v_1,v_2\},\{v_2,v_3\}$ in~$P$ by a path of length at most~$k'$ in~$G_{v_2}$.
	Let~$\{v_p,z\}$ be the last edge in~$P$,~be at position~$j$ on~$z$,~$v \in C_z^j$, and~$v' \in C_z^{j'}$.
	We add a \npath{v,v'} of length~$k'$ in~$G_z - S$ (\cref{claim:planar-np}).
	Note that~$P'$ visits at most~$\ell+1$ vertex-gadgets and~$\ell$~edge-gadgets.
	The length of~$P'$ is at most
	$
		2 + (\ell+1)\cdot k' + \ell \left[ (\ell+1) \cdot k' + 1 \right] = \ell.
	$
	This contradicts that~$S$ forms a solution for~$\I'$.
	It follows that there is no \npath{s,z} of length at most~$\ell$ in~$G-C$ and~$\I$ is a \yes-instance.
	%
	%
\end{proof}

From the proofs of~\cref{thm:nphard,lemma:from-str-to-nonstr} (planarity-preserving reductions for the underlying graph), together with~\cref{lemma:np-hard-planar} we get the following:
\begin{cor}
  \label{cor:sprobhardonplanar}
 Both~\nonstrproblem{} and~\strproblem{} on planar temporal graphs are \NP{}-complete.
\end{cor}

In contrast to the case of general temporal graphs, \strproblem{} on planar temporal graphs is efficiently solvable if the maximum label~$\tau$ is any constant.
To this end, we introduce monadic second-order (\MSO) logic and the optimization variant of Courcelle's Theorem \cite{ARNBORG1991308,courcelle2012graph}.

\emph{Monadic second-order (\MSO) logic} consists of a countably infinite set of (individual) variables, unary relation variables, and a vocabulary.
A \emph{formula} in monadic second-order logic on graphs is constructed from the vertex variables, edge variables, an incidence relation between edge and vertex variables ${\rm inc}(e,x)$, the relations~$=, \land, \lor, \neg$, and the quantifiers~$\forall$ and~$\exists$.
We will make use of many folklore shortcuts, for example: $\not=$, $\subseteq$ , $v \in V$ to check whether $v$ is a vertex variable, and $e \in E$ to check whether $e$ is an edge variable.
The \emph{treewidth} $\tw(G)$ of a graph $G$ measures how tree-like a graph is.
In this paper we are using the following theorem and the treewidth of a graph as a black box.
For  formal definitions of treewidth and \MSO{}, refer to \citet{courcelle2012graph}.
\begin{theorem}[\citet{ARNBORG1991308,courcelle2012graph}]
	\label{thm:mso-opt}
	There exists an algorithm that, given
	\begin{compactenum}[(i)]
		\item an \MSO{} formula $\rho$ with free monadic variables $X_1,\dots,X_r$,
		\item an affine function $\alpha(x_1,\dots, x_r)$, and
		\item a graph $G$,
	\end{compactenum}
	finds the minimum (maximum) of~$\alpha(|X_1|,\dots,|X_r|)$ over evaluations of~$X_1,\dots,X_r$ for which formula~$\rho$ is satisfied on $G$.
	The running time is~$f(|\rho|,\tw(G)) \cdot |G|$, where~$|\rho|$ is the length of~$\rho$, and~$\tw(G)$ is the treewidth of~$G$.
\end{theorem}
We employ \cref{thm:mso-opt} to solve \strproblem{} on planar temporal graphs.
\begin{restatable}{proposition}{planarfpt}
\label{thm:planarfpt}
	\strproblem{} on planar temporal graphs can be solved in $\ON(|\TE|\cdot\log|\TE|)$ time, if the maximum label~$\tau$ is constant.
\end{restatable}
\begin{proof}
	Let $\I = (\TG =(V,\TE,\tau),s,z,k)$ be an instance of \strproblem{}, where the underlying graph of $\TG$ is planar.

	We  define the optimization variant of \strproblem{} in monadic second order logic (\MSO{}) and show that it is sufficient to solve the \strproblem{} on a temporal graph where we can upper-bound the treewidth of the underlying graph by a function only depending on~$\tau$.

	We define the edge-labeled graph $L(\TG)$ to be the underlying graph~$\ug{G}$ of $\TG$ with the edge-labeling $\omega : E(\ug{G}) \rightarrow [2^\tau-1]$ with~$\omega(\{v,w\}) = \sum_{i=1}^{\tau} \indic_{\{v,w\} \in E_i} \cdot 2^{i-1}$, where $\indic_{\{v,w\} \in E_i}=1$ if and only if~$(\{v,w\},i)\in\TE$, and $0$ otherwise.
	Observe that in binary representation, the $i$-th bit of~$\omega(\{v, w\})$ is one if and only if $\{v,w\}$ exists at time point~$i$. 
	To compute~$L(\TG)$ we first iterate over the time-edge set $\TE$ and add an edge $\{v,w\}$ to the edge set $E$ of the underlying graph $L(\TG)$ whenever we find a time-edge~$\big(\{v,w\},t\big) \in \TE$.
	Furthermore, we associate with each edge in $E$ an integer~$\omega_{\{v,w\}}$ which is initially zero, and gets an increment of $2^{t-1}$ for each time-edge~$\big(\{v,w\},t\big) \in \TE$.
	With the appropriate data structure, this can be done in~$\ON(|\TE|\cdot\log|\TE|)$~time. 
	Observe, that in the end $\omega_{\{v,w\}} = \omega(\{v,w\})$. 
	Finally, we copy the vertex set of $\TG$, and hence compute $L(\TG)$ in~$\ON(|V|+|\TE|\cdot\log|\TE|) = \ON(|\TE|\cdot\log|\TE|)$ time.

	\citet{ARNBORG1991308} showed that it is possible to apply \cref{thm:mso-opt} to graphs in which edges have labels from a fixed finite set, either by augmenting the graph logic to incorporate predicates describing the labels, or by representing the labels by unquantified edge set variables. 
	Since \cref{thm:mso-opt} gives a linear-time algorithm with respect to the input size, we can conclude that if the size of that label set depends on~$\tau$, then \cref{thm:mso-opt} still gives a linear-time algorithm for arbitrary but fixed~$\tau$, the size of the \MSO{} formula, and the treewidth.

	We will define the optimization variant of \strproblem{} in \MSO{} on~$L(\TG)$.
	First, the \MSO{} formula ${\rm layer}(e,t) :=~\bigvee_{i=1}^\tau \bigvee_{j\in \sigma(i, 2^\tau-1)} \big(t = i~\land \omega(e) = j\big) $ checks whether an edge $e$ is present in the layer $t$, where $\sigma(i,j) := \set{ x \in [j] ; i\text{-th bit of }x\text{ is 1} }$.
	Note that the length of ${\rm layer}(e,t)$ is upper-bounded by a some function in $2^{\ON(\tau)}$.
	Second, we can write an \MSO{} formula ${\rm tempadj}(v,w,t) :=~\exists_{e \in E}\big({\rm inc}(e,v) \land {\rm inc}(e,w) \land {\rm layer}(e,t)\big)$ to determine whether two vertices~$v$ and~$w$ are adjacent at time point $t$.
	Since the length of ${\rm layer}(e,t)$ is upper-bounded by some function in~$2^{\ON(\tau)}$, the length of ${\rm tempadj}(v,w,t)$ is upper-bounded by some function in~$2^{\ON(\tau)}$ as well.
	Third, there is an \MSO{} formula $${\rm path}(S) := \exists_{x_1,\dots,x_{\tau+1} \in V \setminus S}\Big( x_1 = s \land x_{\tau+1}=z \land \bigwedge_{i=1}^{\tau} \big( x_i = x_{i+1} \lor {\rm tempadj}(x_i,x_{i+1},i) \big)\Big)$$ to check whether there is a \strpath{s,z} which does not visit any vertex in~$S$.
	Observe, that the length of ${\rm path}(S)$ is upper-bounded by some function in~$2^{\ON(\tau)}$.
	Finally, with \cref{thm:mso-opt} we can solve the optimization variant of \strproblem{} by the formula~$\phi(S):=S \subseteq (V\setminus \{ s, z\}) \land \neg {\rm path}(S)$ and the affine function~$\alpha(x) = x$ in $f(\tau,\tw(\ug{G})) \cdot |V|$ time, where~$f$ is some computable function.
	This gives us an overall running time of $\ON(f(\tau,\tw(\ug{G})) \cdot |\TE|\cdot\log|\TE|)$ to decide~$\I$ for \strproblem{}.

	Note that every \strpath{s,z} in $\TG$ and its induced \npath{s,z} in $\ug{G}$ have length at most $\tau$.
	Hence, we can remove all vertices from $\ug{G}$ which are not reachable from~$s$ by a path of length at most $\tau$.
	We can find the set of vertices $N$ which is reachable by a path of length $\tau$ by a breadth first search from~$s$ in linear time.
	Therefore, it is sufficient to solve the \strproblem{} on $\TG[N \cup \{s\}]$.
	Observe, that $L(\TG[N \cup \{s\}])$ has diameter at most $2\cdot\tau$ and is planar.
	The treewidth of a planar graph is at most three times the diameter (see \citet{flum2006parameterized}).
	Thus, (the optimization variant of) \strproblem{}  can be solved in~$f'(\tau) \cdot |\TE|\cdot\log|\TE|$~time, where~$f'$ is a computable function.
	For any constant~$\tau$, this yields a running time of~$\ON(|\TE|\cdot\log|\TE|)$. 
\end{proof}

\section{On Temporal Graphs with Small Temporal Cores}
\label{sec:strict}



In this section, we investigate the complexity of deciding \sproblem{} on temporal graphs where the number of vertices whose incident edges change over time is small. 
We call the set of such vertices the \emph{temporal core} of the temporal graph. 
\begin{definition}[Temporal core]
  For a temporal graph $\TG = (V,\TE,\tau)$, the vertex 
  set~$W=\{v\in V\mid \exists \{v,w\}\in (\bigcup_{i=1}^\tau E_i) \setminus (\bigcap_{i=1}^\tau E_i)\} \subseteq V$ 
  is called the \emph{temporal core}.
\end{definition}
A temporal graph is often composed of a public transport system and an 
ordinary street network.
Here, the temporal core consists of vertices involved in the public transport system.

For~\strproblem{}, we can observe that the hardness reduction described in the proof of \cref{thm:nphard} produces an instance with an empty temporal core. 
%
 In stark contrast, we show that \nonstrproblem{} is fixed-parameter tractable when parameterized by the size of the temporal core\footnote{Note that we can compute the temporal core in $\ON(|\TE|\log|\TE|)$~time.}. 
 We reduce an instance to \multiwaycut{} (NWC) in such a way that we can use an above lower bound \FPT-algorithm due to~\citet{cygan2013multiway} for \multiwaycutAcr{} as a subprocedure in our algorithm for \nonstrproblem{}.
Note that the above lower bound parameterization is crucial to obtain the desired \FPT-running time bound. 
 Recall the definition of \multiwaycutAcr{}:
  \problemdef{\multiwaycut{} (NWC)}
  {An undirected graph~$G=(V,E)$, a set of terminals~$T \subseteq V$, and an integer~$k$. }
  {Is there a set~$S \subseteq (V\setminus T)$ of size at most~$k$ such there is no \npath{t_1,t_2} for every distinct~$t_1,t_2 \in T$?}
  We remark that Cygan et al.'s algorithm can be modified to obtain a solution~$S$.
Formally, we show the following.



\begin{theorem}
	\label{thm:fptcore}
	\nonstrproblem{} can be solved in~$2^{|W| \cdot (\log|W| + 2)}\cdot |V|^{\ON(1)}+\ON(|\TE|\log|\TE|)$ time, where~$W$ denotes the temporal core of the input graph.
\end{theorem}
\begin{proof}
   Let instance~$\I = (\TG = (V,\TE,\tau),s,z,k)$ of \nonstrproblem{} with temporal core~$W\subseteq V$ be given.
  Without loss of generality, we can assume that~$s,z \in W$, as otherwise we add two vertices one being incident only with~$s$ and the other being incident only with~$z$, both only in layer one.
  Furthermore, we need the notion of a \emph{maximal static subgraph}~$\widehat{G}$ of a temporal graph $\TG=(V,\TE)$: 
  It contains all edges that appear in every layer, more specifically~$\widehat{G}=(V,\widehat{E})$ with~$\widehat{E} = \bigcap_{i\in[\tau]}E_i$. 
  Our algorithm works as follows.
  \begin{compactenum}[(1)]
  \item Guess a set~$S_W \subseteq (W \setminus \{ s,z \})$ of size at most~$k$.\label{alg:i}
  \item Guess a number $r$ and a partition~$\{ W_1,\dots,W_r \}$ of~$W \setminus S_W$ such that~$s$ and~$z$ are not in the same~$W_i$, for some~$i \in [r]$.\label{alg:ii}
  \item Construct the graph~$G'$ by copying~$\widehat{G} - W$ and adding a
  vertex~$w_i$ for each part $W_i$. Add edge sets $\{\{v,w_i\} \mid v\in
  N_{\widehat{G}}(W_i)\setminus W\}$ for all $i\in[r]$ and for all $i, j\in[r]$ add
  an edge~$\{w_i,w_j\}$ if $N_{\widehat{G}}(W_i)\cap W_j \neq\emptyset$. 
  \label{alg:iii}
  \item Solve the \multiwaycutAcr{} instance~$\I' = (G', \{w_1,\dots,w_r\},k-|S_W|)$.\label{alg:iv}
  \item If a solution~$S'$ is found for~$\I'$ and~$S' \cup S_W$ is a solution for~$\I$, then output \yes.\label{alg:v}
  \item If all possible guesses in \condRef{alg:i} and \condRef{alg:ii} are considered without finding a solution for~$\I$, then output \no.\label{alg:vi}
  \end{compactenum}
  See \cref{fig:wormhole} for a visualization of the constructed graph $G'$.
  \begin{figure}
\begin{center}
	  \begin{tikzpicture}[yscale=0.66]

	  \usetikzlibrary{calc, patterns}
	  \def\xr{1}
	  \def\yr{1}

	  \begin{scope}
	  \draw[fill=gray!20] (0,0) ellipse (3 and 2);
	  \clip (0,0) ellipse (3 and 2);

	\node at (0.2,1.2)[circle, fill=ourgreen!25!white, minimum height=1, minimum width=1, draw]{\phantom{$kW_2$}};
	\node at (0.2,1.2)[circle, fill=ourgreen!75!white, minimum height=1, minimum width=1, draw]{$W_2$};
	\node at (-2.1,0.25)[circle, fill=lipicsyellow!25!white, minimum height=1, minimum width=1, draw]{\phantom{$kW_1$}}; 
	\node at (-2.1,0.25)[circle, fill=lipicsyellow, minimum height=1, minimum width=1, draw]{$W_1$};
	  \node at (-2.35,0)[circle, fill, draw, scale=1/2, label=-115:{$s$}]{};
	\node at (-0.4,-1.25)[circle, fill=lipicsyellow!25!white, minimum height=1, minimum width=1, draw]{\phantom{$kW_1$}};
	 \node at (-0.4,-1.25)[circle, fill=lipicsyellow, minimum height=1, minimum width=1, draw,label={[label distance=-20]180:{$W_1$}}]{\phantom{$W_1$}};
	 \node at (2.25,0.35)[circle, fill=ourblue!25!white, minimum height=1, minimum width=1, draw]{\phantom{$kW_3$}};
	 \node at (2.25,0.35)[circle, fill=ourblue, minimum height=1, minimum width=1, draw]{$W_3$};
	\node at (2.35,0)[circle, fill, draw, scale=1/2, label=-75:{$z$}]{};
	\node at (0.7,-0.25)[circle, fill=ourblue!25!white, minimum height=1, minimum width=1, draw]{\phantom{$kW_3$}};
	  \node at (0.7,-0.25)[circle, fill=ourblue, minimum height=1, minimum width=1, draw, label={[label distance=-20]0:{$W_3$}}]{\phantom{$W_3$}};
	  \draw[fill=white] (-0.2,-2.1) -- (0.2,-2.1) to [out=90,in=-90] (0.9,2.1) --  (0.5,2.1) to [out=-90,in=45] (0,0.25) to [out=135,in=-90] (-0.5,2.1) -- (-0.9,2.1) to [out=-90,in=90](-0.2,-2) -- cycle;
	  \draw[fill=ourred, fill opacity=0.33] (-0.2,-2.1) -- (0.2,-2.1) to [out=90,in=-90] (0.9,2.1) --  (0.5,2.1) to [out=-90,in=45] (0,0.25) to [out=135,in=-90] (-0.5,2.1) -- (-0.9,2.1) to [out=-90,in=90](-0.2,-2) -- cycle;
	  \begin{scope}
	  \clip (-0.2,-2.1) -- (0.2,-2.1) to [out=90,in=-90] (0.9,2.1) --  (0.5,2.1) to [out=-90,in=45] (0,0.25) to [out=135,in=-90] (-0.5,2.1) -- (-0.9,2.1) to [out=-90,in=90](-0.2,-2) -- cycle;
	  \draw[pattern=north west lines,draw=none] (0,0) circle (3);
   	  \node at (0.2,1.2)[circle, minimum height=1, minimum width=1, draw]{\phantom{$kW_2$}};
	  \node at (0.2,1.2)[circle,minimum height=1, minimum width=1, draw]{$W_2$};
   	  \node at (-0.4,-1.25)[circle, minimum height=1, minimum width=1, draw]{\phantom{$kW_1$}};
	  \node at (-0.4,-1.25)[circle,minimum height=1, minimum width=1, draw]{\phantom{$W_1$}};
   	  \node at (0.7,-0.25)[circle, minimum height=1, minimum width=1, draw]{\phantom{$kW_3$}};
	  \node at (0.7,-0.25)[circle,minimum height=1, minimum width=1,draw]{\phantom{$W_3$}};
	  \end{scope}
	  \end{scope}
	  \node at (0,-2.5)[circle, pattern=north west lines, minimum height=1, minimum width=1, draw]{};
	  \node at (0,-2.5)[circle, fill=ourred, fill opacity=0.33, minimum height=1, minimum width=1, draw, label=0:{$S$}]{};
	  \draw[line width=1pt] (0,0) ellipse (3 and 2);

	  \begin{scope}[xshift=7.5 cm]
	  \draw[fill=gray!20] (0,0) ellipse (3 and 2);
	  \clip (0,0) ellipse (3 and 2);
	\node (w3) at (0.2,1.2)[circle, fill=ourgreen!25!white, minimum height=1, minimum width=1, draw]{\phantom{$kW_2$}};
	  \node at (0.2,1.2)[circle, fill=white, minimum height=1, minimum width=1, draw]{\phantom{$W_2$}};
	 \node (w1)at (-2.1,0.25)[circle, fill=lipicsyellow!25!white, minimum height=1, minimum width=1, draw]{\phantom{$kW_1$}}; 
	\node at (-2.1,0.25)[circle, fill=white, minimum height=1, minimum width=1, draw]{\phantom{$W_1$}};
	 \node (w2) at (-0.4,-1.25)[circle, fill=lipicsyellow!25!white, minimum height=1, minimum width=1, draw]{\phantom{$kW_1$}};
	 \node at (-0.4,-1.25)[circle, fill=white, minimum height=1, minimum width=1, draw]{\phantom{$W_1$}};
	 \node  (w4)  at (2.25,0.35)[circle, fill=ourblue!25!white, minimum height=1, minimum width=1, draw]{\phantom{$kW_3$}};
	 \node at (2.25,0.35)[circle, fill=white, minimum height=1, minimum width=1, draw]{\phantom{$W_3$}};
	\node (w5) at (0.7,-0.25)[circle, fill=ourblue!25!white, minimum height=1, minimum width=1, draw]{\phantom{$kW_3$}};
	  \node at (0.7,-0.25)[circle, fill=white, minimum height=1, minimum width=1, draw]{\phantom{$W_3$}};
	  \begin{scope}
	  \clip (-0.2,-2) -- (0.2,-2) to [out=90,in=-90] (0.9,1.95) --  (0.5,1.98) to [out=-90,in=45] (0,0.25) to [out=135,in=-90] (-0.5,1.98) -- (-0.9,1.95) to [out=-90,in=90](-0.2,-2) -- cycle;
			
	  \node at (0.2,1.2)[circle, pattern=north west lines, minimum height=1, minimum width=1, draw]{\phantom{$W_1$}};
	  \node at (0.2,1.2)[circle, fill=ourred, fill opacity=0.33, minimum height=1, minimum width=1, draw]{\phantom{$W_1$}};
	  \node at (-0.4,-1.25)[circle, pattern=north west lines, minimum height=1, minimum width=1, draw]{\phantom{$W_1$}};
	  \node at (-0.4,-1.25)[circle,fill=ourred, fill opacity=0.33, minimum height=1, minimum width=1, draw]{\phantom{$W_1$}};
	  \node at (0.7,-0.25)[circle, pattern=north west lines, minimum height=1, minimum width=1, draw]{\phantom{$W_1$}};
	  \node at (0.7,-0.25)[circle, fill=ourred, fill opacity=0.33, minimum height=1, minimum width=1, draw]{\phantom{$W_1$}};
	  \end{scope}

	  \node (v2) at (0.2,1.2)[circle, fill=ourgreen!75!white, scale=2/3, draw,label={[label distance=-5]135:{$w_2$}}]{};
	  \node (v1) at (-1.25,-0.5)[circle, fill=lipicsyellow, scale=2/3, draw,label={[label distance=-4]225:{$w_1$}}]{};
	  \node (v3) at (1.4,0.4)[circle, fill=ourblue, scale=2/3, draw,label={[label distance=-2]90:{$w_3$}}]{};

	  \draw (v1) -- (w1.north);
	  \draw (v1) -- (w1.south west);
	  \draw (v1) -- (w1.east);
	  \draw (v1) -- (w2.north west);
	  \draw (v1) -- (w2.west);
	  \draw (v1) -- (w2.south);
	  \draw (v2) -- (w3.north);
	  \draw (v2) -- (w3.west);
	  \draw (v2) -- (w3.south west);
	  \draw (v3) -- (w4.north);
	  \draw (v3) -- (w4.east);
	  \draw (v3) -- (w4.south);
	  \draw (v3) -- (w5.north);
	  \draw (v3) -- (w5.east);
	  \draw (v3) -- (w5.south);

	  \end{scope}
	  
	  \begin{scope}[xshift=7.5 cm]
	    \draw[line width=1pt] (0,0) ellipse (3 and 2);	  
	    \node at (0,-2.5)[circle, pattern=north west lines, minimum height=1, minimum width=1, draw]{};
	    \node at (0,-2.5)[circle, fill=ourred, fill opacity=0.33, minimum height=1, minimum width=1, draw, label=0:{$S_W$}]{};
	  \end{scope}
	  \end{tikzpicture}
	  \end{center}
	  \caption{
	  Illustration of the idea behind the proof of~\cref{thm:fptcore}.
	  \emph{Left-hand side}:
	  Sketch of a temporal graph $\TG$ (enclosed by the ellipse) with temporal~$(s,z)$ separator~$S$ (red hatched) and induced partition~$\{S_W,W_1,W_2,W_3\}$ of the temporal core $W$, where~$S_W=W\cap S$.
	  The outer rings of~$W_1,W_2,W_3$ contain the open neighborhood of the sets.
	  \emph{Right-hand side}:
	  Sketch of the constructed graph $G'$ (enclosed by the ellipse). The partition $\{S_W,W_1,W_2,W_3\}$ is guessed in steps \condRef{alg:i} and \condRef{alg:ii}. 
	  The vertices~$w_1,w_2,w_3$ with edges to the neighborhood of $W_1,W_2,W_3$, respectively, are created in step~\condRef{alg:iii}. 
	  }
	  \label{fig:wormhole}
  \end{figure}
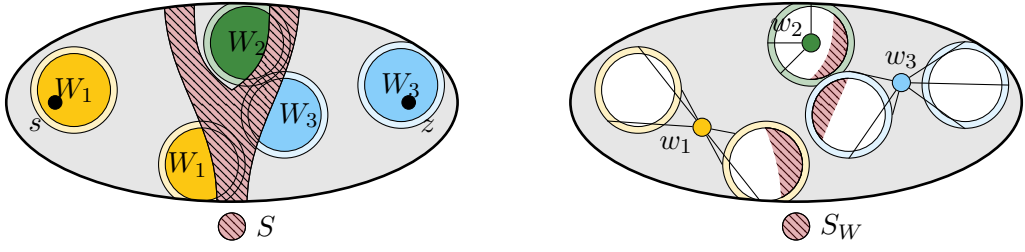
  Since we do a sanity check in step~\condRef{alg:v} it suffices to show that if~$\TG$ has a \nonstrsep{s,z} of size at most~$k$,
  then there is a partition~$\{ S_W, W_1,\dots,W_r \}$ of~$W$ where $s$ and $z$ are in different parts such that
  \begin{inparaenum}[(i)]	
	  \item\label{pr:mwc-hassol} the \multiwaycutAcr{} instance~$\I'$ has a solution of size at most~$k - |S_W|$, and
	  \item\label{pr:anysolworks} if~$S'$ is a solution to~$\I'$, then~$S_W \cup S'$ is a \nonstrsep{s,z} in~$\TG$.
  \end{inparaenum}
  
  Let~$S$ be a \nonstrsep{s,z} of size at most~$k$ in~$\TG$.
  First, we set~$S_W = S \cap W$.
  Let~$C_1,\dots,C_r$ be the connected components of~$\widehat{G} - S$ with $C_i \cap W\neq \emptyset$ for all~$i\in [r]$.
  Now we construct a partition~$\{ S_W, W_1,\dots, W_r \}$ of $W$ such that~$W_i = W \cap C_i$ for all~$i \in [r]$. 
  It is easy to see that $s$ and $z$ are in different parts of this partition. Observe that for~$i,j \in [r]$ with~$i \neq j$ the vertices~$v \in W_i$ and~$u \in W_j$ are in different connected components of~$\widehat{G}-S$.
  Hence,~$w_1,\dots,w_r$ are in different connected components of~$G' - (S \setminus S_W)$.
  Thus~$S \setminus S_W$ is a solution of size at most~$k - |S_W|$ of the \multiwaycutAcr{} instance~$\I' = (G', \{w_1,\dots,w_r\},k-|S_W|)$, proving~\condRef{pr:mwc-hassol}.

  For the correctness, it remains to prove~\condRef{pr:anysolworks}.
  Let~$S'$ be a solution of size at most~$k - |S_W|$ of the \multiwaycutAcr{} instance~$\I'$.
  We need to prove that~$S'\cup S_W$ forms a temporal $(s,z)$-separator in~$\TG$.
  Clearly, if~$S'=S\setminus S_W$, we are done by the arguments before.
  Thus, assume~$S'\neq S\setminus S_W$.
  Since~$S'$ is a solution to~$\I'$, we know that~$w_1,\dots,w_r$ are in different connected components of~$G' - S'$.
  Hence, for~$i, j \in [r]$ with~$i \neq j$ the vertices~$v \in W_i$,~$u \in
  W_j$ are in different connected components of~$\widehat{G}-(S'\cup S_W)$. 

 Now assume towards a contradiction that there is a \nonstrpath{s,z}~$P$ in~$\TG - (S' \cup S_W)$.
  Observe that~$\{s,z\} \subseteq V(P) \cap W$.
  Hence, we have two different vertices~$u_1,u_2 \in V(P) \cap W$
  such that there is no \nonstrpath{u_1,u_2} in~$\TG - S$ and all
  vertices that are visited by $P$ between~$u_1$ and~$u_2$ are contained in $V\setminus W$: Take the furthest vertex in $P$ that is also contained in~$W$ and is reachable by a temporal path from~$s$ in~$\TG - S$ as~$u_1$, and take the next vertex (after~$u_1$) in~$P$ that is also contained in~$W$ as~$u_2$.
  Note that~$u_1$ and~$u_2$ are disconnected in~$\widehat{G}-S$, and hence there are~$i,j\in[r]$ with~$i\neq j$ such that~$u_1\in W_i$ and~$u_2\in W_j$.
  Since~$P$ does not visit any vertices in~$(S'\cup S_W)$ we can conclude
  that~$u_1$ and~$u_2$ are connected in~$\widehat{G} - (S'\cup S_W)$, and hence~$w_i$ and~$w_j$ are connected in~$G'-S'$.
  This contradicts the fact that~$S'$ is a solution for~$\I'$.

  \smallskip\noindent\emph{Running time:}
  It remains to show that the our algorithm runs in the proposed time.
  For the guess in step~\condRef{alg:i} there are at most~$2^{|W|}$ many possibilities.
  For the guess in step~\condRef{alg:ii} there are at most~$B_{|W|} \leq 2^{ |W|\cdot \log(|W|)}$ many possibilities, where~$B_n$ is the $n$-th Bell number.
  Step~\condRef{alg:iii} and the sanity check in step~\condRef{alg:v} can clearly be done in polynomial time.

  Let~$L$ be a minimum \nsep{s,z} in~$\widehat{G} - (W \setminus \{s,z\})$.
  If $k \ge |W \setminus \{s,z\}| + |L|$, then~$(W \setminus \{s,z\}) \cup L$ is a \nonstrsep{s,z} of size at most~$k$ for $\TG$.
  Otherwise, we have that~$k - |L| < |W|$. 
  \citet{cygan2013multiway} showed
  that \multiwaycutAcr{} can be solved in~$2^{k-b}\cdot |V|^{\ON(1)}$ time
, where~$b := \max_{x \in T} \min \set{ |S| ; S\subseteq V  \text{ is an \nsep{x,T \setminus \{x\}}}}$. 
  Since~$s$ and~$z$ are not in the same~$W_i$ for any~$i \in [r]$, we know that~$|L| \leq b$.
  Hence,~$k - b \leq k - |L| < |W|$ and step \condRef{alg:iv} can be done in~$2^{|W|}\cdot |V|^{\ON(1)}$ time.
  Thus we have an overall running time of~$2^{|W| \cdot (\log|W| + 2)}\cdot |V|^{\ON(1)}+\ON(|\TE|\log|\TE|)$.
\end{proof}
We conclude that the strict and the non-strict variant of \nonstrproblem{} behave very differently on temporal graphs with a constant-size temporal core. 
While the strict version stays \NP{}-complete, the non-strict version becomes polynomial-time solvable.

\section{Conclusion}\label{sec:conclusion}
The temporal path 
model strongly matters when assessing the computational complexity of 
finding small separators in temporal graphs. 
This phenomenon has so far been neglected in the literature. 
We settled the complexity dichotomy of \nonstrproblem{} and \strproblem{} by proving \NP{}-hardness on temporal graphs with~$\tau\geq 2$ and~$\tau\geq 5$, respectively, and polynomial-time solvability if the number of layers is below the respective constant.
The mentioned hardness results also imply that both problem variants are \Wone-hard when parameterized by the solution size~$k$. When considering the parameter combination $k+\tau$, it is easy to see that \strproblem{} is fixed-parameter tractable~\cite{Zschoche17}: There is a straightforward search-tree algorithm that branches on all vertices of a \strpath{s,z} which has length at most $\tau$. 
Whether the non-strict variant is fixed-parameter tractable regarding the same parameter combination remains open.
%

We showed that~\sproblem{} on temporal graphs with planar underlying graphs remains \NP-complete.
However, for the planar case we proved that if additionally the number~$\tau$ of layers is a constant, then \strproblem{} is solvable in $\ON(|\TE|\cdot\log|\TE|)$ time.
We leave open whether \nonstrproblem{} admits a similar result.
Finally, we introduced the notion of a temporal core as a temporal graph parameter.
We proved that on temporal graphs with constant-size temporal core, while \strproblem{} remains \NP-hard, \nonstrproblem{} is solvable in polynomial time.

\subparagraph*{Acknowledgements.}
We thank anonymous reviewers for their constructive feedback which helped us to improve the presentation of this work.

\bibliography{tempseparator}


\end{document}